\newif\ifready\readyfalse
\pdfoutput=1

\documentclass[twoside,leqno,twocolumn]{article}
\usepackage[letterpaper]{geometry}
\usepackage{ltexpprt}

\usepackage[utf8]{inputenc}
\usepackage{graphicx}
\usepackage{graphics}
\usepackage{booktabs}
\usepackage{amsmath}
\usepackage{url}
\usepackage{bm}
\usepackage{relsize}
\usepackage[noend,ruled,linesnumbered]{algorithm2e}
\SetKwProg{Fn}{Function}{}{}
\usepackage[noend]{algpseudocode}
\usepackage{multicol}
\usepackage[most]{tcolorbox}
\usepackage{color}
\usepackage{pgfplots}

\usepgflibrary{patterns}
\usepgflibrary{plotmarks}
\definecolor{mydarkpurple}{HTML}{E75480}
\definecolor{mydarkpink}{HTML}{000000}
\definecolor{mydarkred}{HTML}{008080}
\definecolor{mydarkorange}{HTML}{FF8C00}
\definecolor{mydarkblue}{HTML}{04009a}
\definecolor{mybluegreen}{HTML}{30BFBF}
\definecolor{mygreen}{HTML}{03C04A}
\definecolor{mypurple}{HTML}{7F00FF}
\definecolor{myyellow}{HTML}{ebdc78}
\definecolor{babyblue}{rgb}{0.54, 0.81, 0.94}
\definecolor{ballblue}{rgb}{0.13, 0.67, 0.8}
\definecolor{bittersweet}{rgb}{1.0, 0.44, 0.37}
\pgfplotsset{compat=1.3}

\usepackage{subcaption}
\usepackage{tablefootnote}
\pgfplotsset{
    compat=1.3,
    legend image code/.code={
        \draw [#1] (0cm,-0.1cm) rectangle (0.6cm,0.1cm);
    },
}

\usepackage{thm-restate}
\usepackage{enumerate}
\usepackage{enumitem}
\setlist{noitemsep,topsep=0pt,parsep=0pt,partopsep=0pt}
\usepackage{mathtools}
\usepackage{xspace}
\usepackage{hyperref}

\usetikzlibrary{patterns}

\usepackage[capitalize,nameinlink]{cleveref}

\newtheorem{thm}{Theorem}
\newtheorem{definition}[thm]{Definition}

\newcommand{\defn}[1]{\textbf{\textit{#1}}}

\crefname{theorem}{Theorem}{Theorems}
\Crefname{lemma}{Lemma}{Lemmas}
\Crefname{claim}{Claim}{Claims}
\Crefname{observation}{Observation}{Observations}
\Crefname{algorithm}{Algorithm}{Algorithms}
\Crefname{myalgctr}{Algorithm}{Algorithms}
\Crefname{challenge}{Challenge}{Challenges}

\algrenewcommand\algorithmicindent{1em}%

\newcommand{\kc}{$k$-core\xspace}
\mathchardef\mhyphen="2D

\newcommand{\GreedyPP}{\texttt{Greedy++}\xspace}

\newcommand{\GreedySortingPP}{\texttt{GreedySorting++}\xspace}
\newcommand{\PARSortingPP}{\texttt{PaRSorting++}\xspace}
\newcommand{\PGreedyPP}{\texttt{ParallelGreedy++}\xspace}
\newcommand{\PARGreedyPP}{\texttt{PaRGreedy++}\xspace}

\newcommand{\PGreedySortingPP}{\texttt{ParallelGreedySorting++}\xspace}
\newcommand{\FISTA}{\texttt{FISTA}\xspace}
\newcommand{\pkmc}{\texttt{PKMC}\xspace}
\newcommand{\ccg}{\texttt{cCoreG++}\xspace}
\newcommand{\ccexact}{\texttt{cCoreExact}\xspace}
\newcommand{\julienne}{\texttt{Julienne}\xspace}
\newcommand{\ApproxGP}{\texttt{PaRApxGreedy++}\xspace}
\newcommand{\ApproxGS}{\texttt{PaRApxSorting++}\xspace}

\DeclarePairedDelimiter\ceil{\lceil}{\rceil}

\DeclareMathOperator{\poly}{poly}

\definecolor{mygreen}{RGB}{20,140,80}
\definecolor{linkcolor}{RGB}{0,0,230}
\definecolor{mylightgray}{RGB}{230,230,230}
\definecolor{verylightgray}{RGB}{245,245,245}

\newcommand{\etal}[0]{et al.}

\newcounter{myalgctr}

\newtcolorbox{OuterBox}[1][]{%
    breakable,
    enhanced,
    frame hidden,
    interior hidden,
    left=-5pt,
    right=-5pt,
    top=-5pt,
    float=p,
    boxsep=0pt,
    arc=0pt
#1}%

\newtcolorbox{InnerBox}[1][]{%
    enforce breakable,
    enhanced,
    colback=gray,
    colframe=white,
#1}%

\newenvironment{tbox}{
\vspace{0.2cm}
\begin{tcolorbox}[width=\columnwidth,
                  enhanced,
                  boxsep=2pt,
                  left=1pt,
                  right=1pt,
                  top=4pt,
                  boxrule=1pt,
                  arc=0pt,
                  colback=white,
                  colframe=black,
	              breakable
                  ]%
}{
\end{tcolorbox}
}

\newcommand{\tboxhrule}[0]{\vspace{0.1cm} {\color{black} \hrule} \vspace{0.2cm}}

\newenvironment{titledtbox}[1]{\begin{tbox}#1 \tboxhrule}{\end{tbox}}

\newcommand{\core}{k}

\newcommand{\eps}{\varepsilon}

\newcommand{\kest}{\hat{\core}}

\newcommand{\whp}{whp\xspace}

\newcommand{\geom}{\mathsf{Geom}}

\newcommand{\nparam}{\nparaminside}

\newcommand{\nparaminside}{\eps/(8\log^2n)}

\newcommand{\draw}{\geom(\nparam)}

\ifready
\newcommand{\julian}[1]{}
\newcommand{\laxman}[1]{}
\newcommand{\qq}[1]{}
\newcommand{\pat}[1]{}
\newcommand{\qqnote}[1]{{}}
\newcommand{\lnote}[1]{{}}
\else
\newcommand{\julian}[1]{{\color{cyan} Julian: #1}}
\newcommand{\laxman}[1]{{\color{brown} Laxman: #1}}
\newcommand{\qq}[1]{{\color{blue} Quanquan: #1}}
\newcommand{\pat}[1]{{\color{violet} Pat: #1}}
\newcommand{\qqnote}[1]{{\color{magenta}\footnote{\color{blue} Quanquan: #1}}}
\newcommand{\lnote}[1]{{\color{blue}\footnote{\color{brown} Laxman: #1}}}
\fi

\newcommand{\myparagraph}[1]{\smallskip \noindent {\bf #1.}}

\begin{document}
\title{Practical Parallel Algorithms for Near-Optimal Densest Subgraphs on Massive Graphs}
\author{
    Pattara Sukprasert\thanks{Databricks, San Francisco, CA} \thanks{A significant part of this work was done while P.S.\ was a Ph.D.\ candidate and Q.C. Liu was a postdoc at Northwestern University.}
    \and
    Quanquan C. Liu\thanks{Simons Institute at UC Berkeley, Berkeley, CA}
    \and
    Laxman Dhulipala\thanks{University of Maryland, College Park, MD}
    \and
    Julian Shun\thanks{MIT CSAIL, Cambridge, MA}
}
\date{}

\sloppy

\renewcommand{\algorithmicrequire}{\textbf{Input:}}
\renewcommand{\algorithmicensure}{\textbf{Output:}}
\algblock{ParFor}{EndParFor}
\algblock{Input}{EndInput}
\algblock{Output}{EndOutput}
\algblock{ReduceAdd}{EndReduceAdd}

\algnewcommand\algorithmicparfor{\textbf{parfor}}
\algnewcommand\algorithmicinput{\textbf{Input:}}
\algnewcommand\algorithmicoutput{\textbf{Output:}}
\algnewcommand\algorithmicreduceadd{\textbf{ReduceAdd}}
\algnewcommand\algorithmicpardo{\textbf{do}}
\algnewcommand\algorithmicendparfor{\textbf{end\ input}}
\algrenewtext{ParFor}[1]{\algorithmicparfor\ #1\ \algorithmicpardo}
\algrenewtext{Input}[1]{\algorithmicinput\ #1}
\algrenewtext{Output}[1]{\algorithmicoutput\ #1}
\algrenewtext{ReduceAdd}[2]{#1 $\leftarrow$ \algorithmicreduceadd(#2)}
\algtext*{EndInput}
\algtext*{EndOutput}
\algtext*{EndIf}
\algtext*{EndFor}
\algtext*{EndWhile}
\algtext*{EndParFor}
\algtext*{EndReduceAdd}

\maketitle

\begin{abstract}
The densest subgraph problem has received significant attention, both in theory and in practice, due to its applications in problems such as community detection, social network analysis, and spam detection. 
Due to the high cost of obtaining exact solutions, 
much attention has focused on designing approximate densest subgraph algorithms. 
However, existing approaches are not able to scale 
to massive graphs with billions of edges. 

In this paper, we introduce
a new framework that combines approximate densest subgraph algorithms with a pruning optimization. We design new parallel variants of the state-of-the-art sequential \GreedyPP\ algorithm, and plug it into our framework in conjunction with a parallel pruning technique based on $k$-core decomposition to obtain parallel $(1+\eps)$-approximate densest subgraph algorithms. 
On a single thread, our algorithms
achieve $2.6$--$34\times$ speedup over \GreedyPP, and obtain up to $22.37\times$ self-relative parallel speedup on a 30-core machine with two-way hyper-threading. 
Compared with the state-of-the-art parallel algorithm by Harb et al.
[NeurIPS'22]
, we achieve up to a $114\times$ speedup on the same machine. 
Finally, against the recent sequential algorithm of Xu et al.
[PACMMOD'23]
, we achieve 
up to a $25.9\times$ speedup.
The scalability of our algorithms enables us to obtain near-optimal density statistics on the \texttt{hyperlink2012} (with roughly 113 billion edges) and \texttt{clueweb} (with roughly 37 billion edges) graphs  for the first time in the literature.

\end{abstract}

\section{Introduction}

The densest subgraph problem is a fundamental problem in graph mining 
that has been studied extensively for decades, both because of its theoretical challenges
and its practical importance. 
The numerous applications of the problem include
community detection and visualization in social networks~\cite{Alvarez2005,CHKZ03,GJLS13,JXR09,KRRT99,RTG14}, 
motif discovery in protein and DNA~\cite{DHZ22,Fratkin06,SSABCHMM15}, and pattern identification~\cite{AKSSST14,DJDLT09,fraudar16}.

Significant effort has been made in the theoretical computer science community in computing 
exact and approximate densest subgraphs under various models of computation, in particular in the 
static~\cite{greedyPP,Charikar00,CQT22,KS09,TG15}, streaming~\cite{BHNT15}, distributed~\cite{BGM14,Ghaffari2019,SuVu20}, 
parallel~\cite{bahmani2012densest,dhulipala2018theoretically,DCS17,harb2022faster}, dynamic~\cite{BHNT15,ChekuriQ22dynamic,CHHRS22,SW20}, 
and privacy-preserving~\cite{DLRSSY22,AHS21,NV21} settings. 
However, despite a plethora of theoretical 
improvements on these fronts, there still does not exist practical near-optimal
densest subgraph algorithms that can scale up to the largest publicly-available graphs 
with tens to hundreds of billions of edges. 
In particular, for the largest such graphs, \texttt{hyperlink2012} (with roughly 113 billion edges) and \texttt{clueweb} (with roughly 37 billion edges), no previous
approximations for the densest subgraph were known that are better than a $2$-approximation. 

There are two typical approaches for solving the densest subgraph problem exactly. The first is 
to solve a combinatorial optimization problem using a linear program solver. The other 
is to set up a flow network with size polynomial in the size of the original graph, and then run a maximum flow algorithm on it. However, 
the caveat to both approaches is that they are not scalable to  modern massive graphs; namely,
both approaches have large polynomial runtimes and the best theoretical algorithms for 
these approaches are often not practical. %
Because of this bottleneck, many have instead
investigated approaches for approximate densest subgraphs. 

The best-known approximation algorithms
for the densest subgraph problem fall into two categories. 
The first category contains parallel approximation algorithms, which work by iteratively removing  carefully chosen subsets of low-degree vertices while computing the density of the induced subgraph of the remaining 
vertices; then, the induced subgraph with the largest density is taken as the 
approximate densest subgraph~\cite{bahmani2012densest,BHNT15,Charikar00} 
using $\poly(\log n)$ rounds of peeling vertices with degree smaller than some threshold.
Unfortunately, such 
methods give $(2+\eps)$-approximations at best and no one has thus far made such methods
work in $\poly(\log n)$ rounds and give better approximations. 

The second category consists of algorithms obtained from the \emph{multiplicative weight update (MWU)} method. 
The multiplicative weight update framework approximately solves an optimization problem by 
using expert oracles to update the weights assigned to the variables multiplicatively and iteratively
over several rounds depending on how the experts performed in previous rounds. The MWU framework 
allows for obtaining $(1+\eps)$-approximate densest subgraphs in $\poly(\log n)$ iterations; however, 
it requires more work than the peeling algorithm per iteration to update the weights of the variables. As such, neither approach 
is particularly scalable to massive graphs. 

In terms of practical solutions, Boob \etal~\cite{greedyPP} present
 a fast, sequential, iterative peeling algorithm called \GreedyPP\ that combines peeling with the MWU framework. 
Chekuri et al.~\cite{CQT22} show that running \GreedyPP\ for 
$\Theta(\frac{\Delta \log n}{\rho^* \eps^2})$ iterations results in a $(1+\eps)$-approximation of the densest subgraph, where $\rho^*$ is the density of the densest subgraph.
However, \GreedyPP\ is not parallel, and does not take advantage of modern multi-core and multiprocessor
architectures. Recently, Harb \etal~\cite{harb2022faster} proposed an iterative algorithm based
on \emph{projections} that solves a quadratic objective function with linear constraints derived from the 
dual of the densest subgraph linear program of Charikar~\cite{Charikar00}. For a graph with $m$ edges and maximum degree $\Delta$, they prove that their algorithm
converges to a $(1+\eps)$-approximation in $O(\sqrt{m \Delta}/{\eps})$ iterations,
where each iteration takes $O(m)$ work. 

Xu \etal~\cite{xu2023efficient} recently
introduce a framework for a generalized version of the densest 
subgraph problem that includes variants like the densest-at-least-$k$-subgraph problem. Their framework alternates between
iteratively using maximum flow to obtain denser subgraphs and then
peeling according to the $k$-core 
to shrink the graph for the next 
maximum flow iteration. However, their algorithm is not parallel and,
thus, cannot scale to the largest publicly available graphs. 
Parallel implementations exist that give $2$-approximations on the 
densest subgraph~\cite{dhulipala2017julienne,dhulipala2018theoretically,luo2023scalable},
but such algorithms and implementations achieve 
worse theoretical approximation guarantees than our 
$(1+\eps)$-approximation algorithms. We also 
demonstrate that they obtain worse empirical approximations.

In our work, we design fast practical algorithms that simultaneously make use of parallelism as well as the closely related concept of the 
\emph{$k$-core decomposition}. The $k$-core decomposition decomposes the graph into \emph{$k$-cores} 
for different values of $k$. Within the induced subgraph of each $k$-core, each vertex has degree at least $k$.
It is a well-known fact that the density of the densest subgraph is within a factor of $2$ of the 
maximum core value. However, it is less clear how to make use of this fact in creating scalable algorithms for 
the largest publicly-available graphs. In this paper, we design a pruning framework that, combined with 
our parallel densest subgraph subroutines, results in both theoretical as well as practical improvements
over the state-of-the-art. The main idea of our framework is to iteratively prune the graph using lower bounds on the density of densest subgraph computed from our parallel  densest subgraph subroutines, while preserving the densest subgraph. 

The concept of using pruning to obtain a smaller subgraph from which to approximate the densest subgraph is also used in some recent works~\cite{FangYCLL19, xu2023efficient}. 
However, in their works, the pruning procedures 
they use are \emph{inherently sequential}.
Compared to previous work, we introduce a \emph{parallel, iterative} 
pruning approach in this paper
and demonstrate via our comprehensive experimentation
that our algorithms are more efficient and more scalable than
all previous baselines.

    Specifically, we give parallel peeling-based MWU and
    sorting-based MWU iterative algorithms that use pruning and 
    are based on \GreedyPP~\cite{greedyPP}. Our algorithms 
    achieve the same theoretical number of iterations as Chekuri \etal~\cite{CQT22}, but is more amenable to parallelization.
    Experimentally, on an $30$-core machine with hyperthreading, our parallel sorting-based algorithm outperforms our parallel peeling-based algorithm, as well as previous state-of-the-art algorithms on most graphs. 
    For instance, compared with the state-of-the-art parallel algorithm by Harb et al.~\cite{harb2022faster}, we achieve up to a $114\times$
    speedup on the same machine.

    Leveraging the scalability of our parallel algorithms, we provide a number of previously unknown graph statistics and 
    graph mining results on the largest of today's publicly available graphs, \texttt{hyperlink2012} and \texttt{clueweb}, using
    commodity multicore machines. We also provide statistics (such as the 
    empirical \emph{width}) that may prove to be interesting and useful
    in aiding future work on this topic.

\section{Preliminaries}
\label{sec:prelim}

Given an undirected, unweighted graph $G=(V,E)$, 
let $n = |V|$ and $m= |E|$.
Let $\deg_G(v)$ be the degree of vertex $v$ in $G$. 
We define the \defn{density} of $G$ to be $\rho(G) = \frac{|E|}{|V|}$. 
The goal of the \defn{densest subgraph} problem is to find a subgraph $S \subseteq G$, such that $\rho(S)$ is maximized. We will use $S^*$ to denote a densest subgraph of $G$ with maximum density $\rho^*$.

    \begin{table}
    \begin{center}
    \footnotesize
    \begin{tabular}{cc}
    \toprule
    Symbol & Meaning \\
    \midrule
    $G=(V,E)$ & undirected, unweighted input graph\\
    \hline
    $n,m$ & number of vertices, edges resp. \\
    \hline
    $\deg(v)$ & current degree of vertex $v$ \\
    $\Delta$ & current maximum degree of graph\\
    \hline
    $c_p$ & Peeling complexity\\
    \hline
    $\rho(G)$ & current density of graph $G$\\
    $\rho^*$ & maximum induced subgraph density of graph $G$\\
    $\tilde \rho$ & the best density found in our algorithms \\
    \hline
    $core(G,k)$ & $k$-core of $G$\\
    $core(v)$ & core number of $v$\\
    $k_{max}$ & max non-empty core number\\
    \hline
    $\ell(v)$ & current load of vertex $v$\\
    \hline
    \end{tabular}
    \end{center}
    \caption{Common notation used throughout the paper. }
    \label{tab:notations}
    \end{table}

A central structure that we study is the $k$-core of an undirected graph. We now define $k$-core formally.

\begin{definition}[$k$-core]
    A $k$-core $core(G,k)$ of $G$ is defined to be a maximal vertex-induced subgraph $S \subseteq G$ such that $\deg_{S}(v) \geq k$ for any $v \in V(S)$.
\end{definition}

It is well known that to find $core(G,k)$, one can repeatedly \emph{peel}\footnote{Throughout this paper, we say a vertex $v$ is \emph{peeled} from $G$ when $v$ and all its adjacent edges are deleted.}
an arbitrary vertex $v$ from $G$ so long as $\deg_G(v) < k$.
This process terminates when all remaining vertices have degree at least $k$, or the graph becomes empty.
If the remaining graph is not empty, then it is the unique subgraph, $core(G,k)$.
Next, we define the {\em coreness} or {\em core number} of a vertex $v$:

\begin{definition}[Core number]
    For any vertex $v$, we let $core(v) = k$ if $k$ is the maximum integer such that $v$ is in $core(G,k)$.
\end{definition}

An easy modification of the peeling algorithm described above yields $core(v)$ for all vertices $v$.
We call this peeling-based algorithm \defn{Coreness}.
In this algorithm, we pick a vertex with minimum degree and peel it one at a time until there are no vertices left.
Let $D$ be a variable that represents the maximum degree of peeled vertices at the time we peel them. 
Initially, $D=0$.
Once $v$ is about to be peeled, we set $D \leftarrow \max(D,\deg_G(v))$. 
We then set $core(v) \leftarrow D$ and peel $v$ from $G$.
We refer to the ordering of vertices that we peel in this process as a \defn{degeneracy ordering} of the graph, which is unique up to permuting vertices in the order with the same coreness.

We also use the following notion of $c$-approximate $k$-core decomposition, which can be computed more efficiently than exact $k$-core. 

\begin{definition}[\unboldmath{$c$}-Approx \unboldmath{$k$}-Core Decomposition]\label{def:approx-k-core}
    A \defn{$c$-approximate} \defn{\kc decomposition}
    is a partition of vertices into layers, such that a vertex $v$ is
    in approximate core $\kest(v)$, denoted $apxcore(G, \kest(v))$, only if
    $\frac{k(v)}{c} \leq \kest(v) \leq ck(v)$, where $k(v)$ is the coreness of $v$.
\end{definition}

Later on, we will want to find an ordering that is similar to the degeneracy ordering, but certain {\em loads} of vertices are also given as input.
Let $\ell(v)$ be {\em load} of $v$. 
At each step, we peel the vertex that minimizes the term $\ell(v) + \deg_G(v)$. Note that after $v$ is peeled, the induced degrees $\deg_G(v')$
of $v$'s neighbors $v'$ are decreased. 
As a special case, we obtain the degeneracy ordering by setting $\ell(v) = 0$ for all $v$. 
For the remainder of this paper, we refer to the ordering obtained using
$\ell(v) + \deg_G(v)$ as the \defn{load ordering}.

\myparagraph{Model Definitions}
We analyze the theoretical efficiency of our parallel algorithms in the \defn{work-depth} model~\cite{CLRS,JaJa92}.
In this model, the \defn{work} is the total number of operations executed by the algorithm and the \defn{depth} (parallel time) is the longest chain of sequential dependencies.
We assume that concurrent reads and writes are supported in $O(1)$ work/depth.
A \defn{work-efficient} parallel algorithm is one with work that asymptotically matches the best-known sequential time complexity for the problem. All of our algorithms presented in this paper are work-efficient.
We say that a bound holds \defn{with high probability (\whp{})} if it holds with probability at least $1-1/n^c$ for any $c\geq1$. 

We use the following parallel primitives in our algorithms:
\defn{ParFor}, \defn{SuffixSum}, \defn{FindMax}, \defn{Bucketing}, 
and \defn{IntegerSort}.
Each primitive takes a sequence $A$ of length $n$. 
\defn{ParFor} is a parallel version of a for-loop that we use
to apply a function $f$ to each element in the sequence.
If a function $f$ takes $O(t)$ work and $O(d)$ depth, then \defn{ParFor} takes $O(tn)$ work and $O(d)$ depth.
\defn{SuffixSum} returns a sequence $B$ where $B[j] = \sum_{i=j}^n A[i]$.
\defn{FindMax} returns an element with maximum value among those in the sequence. 
\defn{SuffixSum} and \defn{FindMax}
can be implemented to take $O(n)$ work and $O(\log n)$ depth.
\defn{IntegerSort} returns a sequence in sorted order (either non-increasing or non-decreasing order) according to integer keys.
We use two different implementations of \defn{IntegerSort}:
the first is an algorithm by Raman~\cite{raman1991} which takes $O(n \log \log n)$ expected work and $O(\log n)$ depth \whp{}, and the second is
a folklore algorithm that takes $O(n/\eps)$ work and $O(n^\eps)$ depth for $0<\eps<1$~\cite{vishkin2010thinking}.
The decision to use one of these two sorting algorithms depends on whether work or depth is more important.
We state the complexity of our algorithm in both ways when necessary.

\subsection{Pruning with Cores}
\label{sec:pruning}
In this section, we describe a pruning idea that takes an input graph $G$ and outputs a subgraph $H \subseteq G$ such that (1) $H$ is smaller than $G$ and (2) any densest subgraph $S^* \subseteq G$ is in $H$.
We begin with a property that relates a graph's density and its vertices' degrees. 

\begin{lemma}[Folklore, (see, e.g., \cite{CQT22})] %
\label{lem:deg_rho}
    Given $G = (V, E)$, if there is a vertex $v$ with degree $\deg_G(v) < \rho(G)$, then $G' = G \setminus \{v\}$ is a graph with density $\rho(G') > \rho(G)$.
\end{lemma}

\begin{proof}
It holds that 
\begin{align*}
\rho(G) &= \frac{|E(G')| + \deg_G(v)}{|V|} \\
    &= \frac{|V|-1}{|V|} \cdot \frac{|E(G')|}{|V|-1} + \frac{\deg_G(v)}{|V|} \\
    &= \frac{|V|-1}{|V|} \cdot \rho(G') +  \frac{\deg_G(v)}{|V|}\\
    &= x \cdot \rho(G') + (1-x) \cdot \deg_G(v),
\end{align*}
for some real number $x \in (0,1)$.
The last line can be viewed as a weighted average between $\rho(G')$ and $\deg_G(v)$.
Since $\deg_G(v) < \rho(G)$, it has to be the case that $\rho(G') > \rho(G)$ so that their average becomes $\rho(G)$. 
\end{proof}

As a corollary, any vertex in a densest subgraph has induced degree at least $\rho(S^*)$.

\begin{corollary}
\label{col:deg_rho}
    Let $S^*$ be the densest subgraph. Then for any $v \in V(S^*)$, $\deg_G(v) \geq \deg_{S^*}(v) \geq \lceil \rho^* \rceil$. 
\end{corollary}

\cref{col:deg_rho} follows immediately from~\cref{lem:deg_rho}
because vertices $v$ with $\deg_{G}(v) < \rho^*$ can be
peeled while increasing the density of the remaining subgraph.
Thus, a natural procedure that we have for obtaining the densest 
subgraph is to iteratively remove any vertex $v$ that 
has degree less than the current density of the subgraph.
Notice that this process is very similar to the algorithm for computing $core(G,k)$ described in \cref{sec:prelim}.
In fact, we can relate $k$-core to the densest subgraph.

\begin{lemma}
\label{lem:core_rho}
For some $k \leq \lceil \rho^* \rceil$, let $C = core(G,k)$ be the $k$-core of $G$. It must be the case that $S^* \subseteq C$.
\end{lemma}

\begin{proof}
We prove this by contradiction. Assume that $S^* \setminus C$ is non-empty (i.e., there is a vertex in $S^*$ but not in $C$). Let $H = S^*  \cup  C$. Notice that, for any vertex $v \in S^* \cup C$, it holds that
$\deg_H(v) \geq k$---if $v \in S^*$, then $\deg_H(v) \geq \deg_{S^*}(v) \geq \ceil{\rho^*} \geq k$ by~\cref{col:deg_rho}, and if $v \in C$, then $\deg_H(v) \geq \deg_C(v) \geq k$. Hence, $S^* \cup C$ is a $k$-core with more vertices than $C$, implying that $C$ is not maximal, which is a contradiction.
\end{proof}

\begin{corollary}[Folklore]\label{cor:kmax}
    Let $k_{max}$ be the maximum integer such that the $core(G,k_{max})$ is not empty. Let $C$ be the $\lceil \frac{k_{max}}{2} \rceil$-core. Then $S^* \subseteq C$. 
\end{corollary}
\begin{proof}
    For any $v$ in $C = core(G, k_{max})$, we have $\deg_S(v) \geq k_{max}$. Hence, 
    $$\rho(C) = |E(C)| / |V(C)| \geq \frac{\sum_{v \in V_C}\deg_C(v)/2}{|V(C)|} \geq k_{max}/2.$$
    Thus, $\rho^* \geq \rho(C) \geq k_{max}/2.$ It follows from \cref{lem:core_rho} that $S^*$ is contained in the $\lceil \frac{k_{max}}{2} \rceil$-core.
\end{proof}

Similarly, the largest non-empty $c$-approximate $k$-core, $\hat{k}_{max}$, also gives us a lower bound on $\rho^*$, in terms of the density of a (potentially larger) approximate core with smaller approximate core number: %

\begin{corollary}\label{cor:apxkmax}
    Let $\hat{k}_{max}$ be the maximum integer such that the $apxcore(G,\hat{k}_{max})$ is not empty. Let $C$ be the $\lceil \frac{\hat{k}_{max}}{2c} \rceil$-approximate core. Then $S^* \subseteq C$. 
\end{corollary}
\begin{proof}
The proof is identical to that of \cref{cor:kmax}; the only difference is that
since the core is approximate, the lower bound on $\deg_{S}(v)$ for any
$v$ in $apxcore(G, \hat{k}_{max})$, is $\hat{k}_{max}/c$.
\end{proof}

\begin{figure}
  \includegraphics[width=\columnwidth]{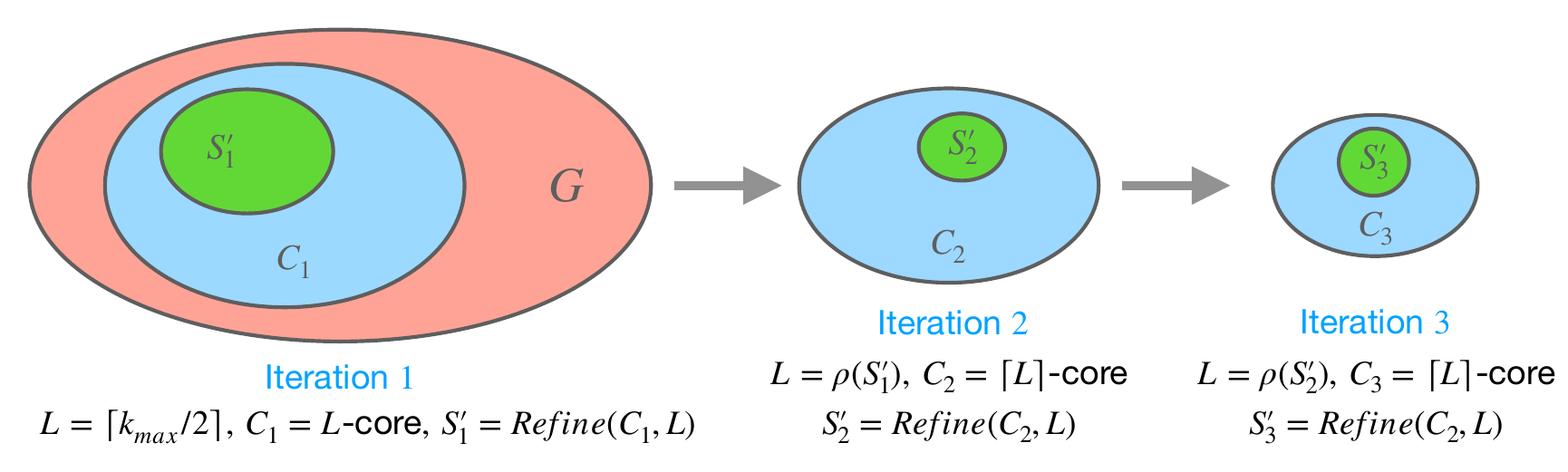}
  \caption{\small Example illustrating the Pruning-and-Refining Framework (Algorithm~\ref{alg:prune-and-iterate}).
  The $i$-th iteration of the algorithm computes a lower bound $L$ on the density, computes the $C_i = \lceil L \rceil$-th core of $G$, and then applies a $Refine$ algorithm on $C_i$ to compute a new subgraph $S'_{i}$. In the example, the density of each successive $S'_{i}$ is increasing, and the cores $C_i$ decrease in size.
  \label{fig:par}}
\end{figure}

\section{Pruning-and-Refining Framework}
Based on the properties described in \cref{sec:prelim}, any algorithm that yields
a lower bound on $\rho^*$ can be used for pruning the graph while retaining the densest subgraph. 
The main idea of the framework is as follows.
Let $L$ be a lower bound on $\rho^*$. 
We can prune the input graph $G$ by computing $G'$, which is the $\lceil L \rceil$-core of $G$ and then search for the densest subgraph in $G'$ instead of $G$. 
This process can be repeated multiple times, 
making it useful in algorithms that iteratively refine (tighten) the
lower bounds for $\rho^*$ over a sequence of steps.

To the best of our knowledge, the idea of using cores
to prune the graph adaptively while refining the approximate densest subgraph solution has not been done in the literature.
The closest idea is from Fang et al.~\cite{FangYCLL19} and Xu et al.~\cite{xu2023efficient}.
In~\cite{FangYCLL19}, their pruning rules first compute \defn{Coreness}, and then inspect connected components  from the $\lceil \frac{k_{max}}{2}\rceil$-core. They take the maximum density found among the connected components as a lower bound and use $k_{max}$ as an upper bound. 
They then run a flow-based algorithm on each connected component separately. Note that flow-based algorithm can only tell if a graph has a subgraph of a specific density $\tilde{\rho}$, so a binary search over the optimal density is required to solve the densest subgraph problem.
Their pruning rules do not help much if there is only a single component in the $\lceil \frac{k_{max}}{2}\rceil$-core. 
Then, in~\cite{xu2023efficient}, they give a sequential pruning-based algorithm based on flow where their
implementation prunes the graph at the beginning and runs a flow-based algorithm to find approximate densities.

\subsection{Framework Overview}
We apply this idea in an algorithmic framework for computing an 
approximate densest subgraph, which is shown in \cref{alg:prune-and-iterate}.
The pseudocode uses {\em exact pruning}, i.e., it uses the value of the exact $k_{max}$-core, but
we also describe how to use approximate $k$-cores below.
On Lines 1--4, we compute the lower bound $L$ by applying \cref{cor:kmax} and either an exact $k$-core algorithm or an approximate $k$-core algorithm. 
Both algorithms take $O(m + n)$ work, but approximate $k$-core has provably poly-logarithmic depth. 
For exact $k$-core, we use the bucketing-based $k$-core implementation of \cite{dhulipala2017julienne,dhulipala2018theoretically}.
The algorithm iteratively peels all vertices with degree at most $d$ in parallel, starting with $d=0$, and incrementing $d$ whenever there are no more vertices with degree at most $d$. The algorithm 
takes $O(m+n)$ expected work and $O(c_p\log n)$ depth with high probability,
where $c_p$ is the \defn{peeling complexity}, which is defined as the number of iterations needed to completely peel the graph.
For approximate $k$-core, we use the implementation of Liu et al.~\cite{LSYDS22}, which gives a $(2+\delta)$-approximation to all core numbers and takes $O(m+n)$ expected work and $O(\log^3n)$ depth \whp{}.
Line 2 shows the lower bound $L$ given $k_{max}$, but we can alternatively compute the lower bound for the approximate $k$-core approach using \cref{cor:apxkmax}.
Given the coreness values in $cores$, we extract the $j$-core from $G$ using $getCore(G,cores,j)$ on Line 3.

On Lines 5--10, we iterate for $T$ rounds, where each round calls a function $Refine$ 
which computes subgraphs with potentially higher density. Note that for algorithms that we use in our paper, our $Refine$ step on Line 6 is oblivious to the parameter $L$. However, knowing $L$ might be useful for other algorithms (e.g., flow-based algorithms). After each refinement step, we then get a potentially better solution, which we memorize in Line 7--9. 
In line 10, we leverage this lower bound $L$ by shrinking $G$ to be $\lceil L \rceil-$core. By \cref{lem:core_rho}, the densest subgraph $\lceil L \rceil$-core contains densest subgraph $S^*$. 
We then return the approximate densest subgraph on Line 11. In \cref{sec:algo}, we describe various options for the $Refine$ function.

\begin{algorithm}[t]\caption{Pruning-and-Refining Framework}\label{alg:prune-and-iterate}
    \SetKwInOut{Input}{Input}
    \SetKwInOut{Output}{Output}
    \SetKwFor{For}{for}{do}{endfor}
    \SetKw{KwIn}{in}
    \SetKw{KwLet}{let}
    \DontPrintSemicolon
    
    \Input{an input graph $G=(V,E)$, number of iterations $T$}
    \Output{an approximate densest subgraph $S$}
    \SetAlgoVlined
    $cores, k_{max} \leftarrow Coreness(G)$\;
    $L \leftarrow \lceil k_{max}/2 \rceil$ %
    
    $G \leftarrow getCore(G, cores, L)$ %
    
    $S \leftarrow G$ \tcp*{Initial pruning}

    \For{i=1 \KwTo T}{
        $S' \leftarrow Refine(G, L)$ \tcp*{Refine candidate subgraph}
        \If{$\rho(S') > \rho(S)$}{ 
            $S \leftarrow S'$\;
            $L \leftarrow \max(L, \rho(S))$ \;
            $G \leftarrow getCore(G, cores, \lceil L \rceil)$ 
        }
    }
    \Return {S}
\end{algorithm}

\subsection{Refinement Algorithms}\label{sec:algo}
Next, we describe algorithms that can be used for the $Refine$ function in \cref{alg:prune-and-iterate}. 
We first describe two existing sequential algorithms, the peeling algorithm and \GreedyPP, and then introduce our parallel algorithms.

\myparagraph{Peeling Algorithm~\cite{Charikar00}} %
At each step, we compute the density of the current graph. Then, we pick a vertex $v$ with the minimum induced degree and remove it from the graph. We continue until there are no vertices remaining, and return the subgraph with the maximum density found in this process. 
The peeling algorithm can be parallelized~\cite{dhulipala2017julienne}, but can have linear depth in the worst case.
Charikar~\cite{Charikar00} proves that the subgraph returned by this peeling algorithm has a density at least half the optimum density, i.e., it gives a $2$-approximation to the densest subgraph.

\myparagraph{\GreedyPP~\cite{greedyPP}} %
\cref{alg:greedypp} presents a greedy load-based densest subgraph algorithm, for which the state-of-the-art
\GreedyPP\ algorithm is a special case. 
Initially, each vertex $v$ is associated with a load $\ell(v) = 0$ (Lines 2--3). 
The algorithm runs for $T$ iterations (Lines 4--11).
On each iteration, we compute the degeneracy order $O$ with respect to the load $\ell$ on graph $H$ to obtain the ordered set of vertices $v_1, \ldots, v_n$ (Line 6).
Then, on Lines 7--11, we peel vertices in this order. When $v_i$ is peeled, we compare the density of the remaining subgraph to the density of the best subgraph found so far, and save the denser of the two. 
We also update the loads, setting $\ell(v_i) \leftarrow \ell(v_i) + \deg^*_H(v_i)$,
where $\deg^*_H(v_i)$ here is the induced degree of $v_i$ when it is peeled.
We return the best subgraph found after $T= \Theta(\frac{\Delta \log n}{\rho^* \eps^2})$ iterations, where
 $\Delta$ is the maximum degree and $0 < \eps < 1$ 
 is an adjustable parameter.
This algorithm yields a $(1+\eps)$-approximation of the densest subgraph as shown in~\cite{CQT22}.
The first iteration of \GreedyPP is exactly the peeling algorithm of Charikar. 

The original \GreedyPP\ algorithm  is implemented in a way where the  degeneracy ordering and the update steps are fused together.
It will become clear once we introduce our algorithm below why we decouple these two steps for obtaining greater parallelization.

\begin{algorithm}[t!]\caption{Greedy Load-Based Densest Subgraph}\label{alg:greedypp}
    \SetKwInOut{Input}{Input}
    \SetKwInOut{Output}{Output}
    \SetKwFor{For}{for}{do}{endfor}
    \SetKw{KwIn}{in}
    \SetKw{KwLet}{let}
    \DontPrintSemicolon
    
    \Input{an input graph $G=(V,E)$, number of iterations $T$, ordering function $O$}
    \Output{an approximate densest subgraph $S$}
    \SetAlgoVlined
    \For{v \KwIn V}{
        $\ell(v) \leftarrow 0$ \;
    }

    $H=(V_H, E_H) \leftarrow (V, E)$\;
    $S \leftarrow G$\;
    \For{i=1 \KwTo T}{
        \KwLet $v_1, \ldots, v_n$ be the ordering provided by the function $O$\;
        \For{j=1 \KwTo n}{
            \If{$\rho(H) > \rho(S)$}{
                $S \leftarrow H$ \;
            }
            $\ell(v_j) \leftarrow \ell(v_j) + \deg^*_H(v_j)$\; 
            \tcp*{$\deg^*_H(v_j)$ is degree of $v_j$ when peeled}
            $H \leftarrow H\setminus \{v_j\}$\;
        }
    }

    \Return{$S$}
        
\end{algorithm}

\myparagraph{\GreedySortingPP~(our algorithm)} 
Our second algorithm uses a simpler method for computing $O$ than \GreedyPP, in that it orders vertices based on their loads at the \emph{beginning} of the iteration. 
The motivation for this algorithm is that sorting is highly parallelizable and is also faster in practice than the iterative peeling process used in \GreedyPP (which has linear depth). 
Therefore, on Line 6 of \cref{alg:greedypp}, we compute $v_1, \ldots, v_n$, such that $\ell(v_1) \leq \ell(v_2) \leq \ldots \leq \ell(v_n)$. 
Because of the way we decouple the ordering and update steps in \GreedyPP, Line 6 is the only difference between the two algorithms.
Next, we argue that \GreedySortingPP\ has the same guarantees in terms of the approximation and number of rounds as \GreedyPP.

\begin{thm}
\label{thm:sortingpp_work}
For $T = \Theta\left( \frac{\Delta \log n}{\rho^* \eps^2} \right )$, \GreedySortingPP\ outputs a $(1+\eps)$-approximation to the densest subgraph problem.
\end{thm}
\begin{proof}
The proof follows almost immediately from Section~4 of \cite{CQT22}.
To prove that \GreedyPP~works, they define an exponential-sized linear program where each variable corresponds to one possible peeling order (i.e., a permutation).
They then utilize the multiplicative weight update (MWU) framework on the linear program.\footnote{See, e.g., \cite{arora2012multiplicative} for a survey on this topic.}
By the way the formulate their linear program, the subproblem that we need to solve w.r.t. MWU is to find a good ordering.
Lemma~4.6 of \cite{CQT22} shows that the ordering obtained with \GreedyPP~is a \emph{good approximate ordering}.
The proof of Lemma~4.6 work for any order $v_1, \dots, v_n$ 
that satisfies the following property:
$\ell(v_i) \leq \ell(v_j) + \Delta$ if $i<j$.
This is true for the ordering used in \GreedySortingPP, where we sort by the initial load of the vertices.
Hence, by plugging in this ordering, all of the proofs in \cite{CQT22} go through.
\end{proof}

\subsection{Parallel Implementation}

In this subsection, we present our parallelizations of \GreedyPP~and \GreedySortingPP. We still run both algorithms for $T$ iterations, one iteration at a time, and our aim is to achieve low depth within each iteration.

\begin{algorithm}[!t]\caption{Parallel Density and Load Computation}\label{alg:peel_ordering}
    \SetKwInOut{Input}{Input}
    \SetKwInOut{Output}{Output}
    \SetKwFor{For}{for}{do}{endfor}
    \SetKwFor{ParFor}{parfor}{do}{endparfor}
    \SetKw{KwIn}{in}
    \SetKw{KwLet}{let}
    \DontPrintSemicolon
    
    \Input{an input graph $G=(V,E)$, an ordering $v_1, \ldots v_n$, current loads $\ell(\cdot)$}
    \Output{updated loads $\ell(\cdot)$, best density found $\rho_{max}$}
    \SetAlgoVlined
    $A :=$ array of size $n$\;
    \ParFor{$e = (v_j, v_k)\ \KwIn\ E$}{
        $A[\min (j,k)] \leftarrow A[\min (j,k)]+1$\;
    }
    $B := SuffixSum (A)$\;
    \ParFor{$i = 1 \ \KwTo\ n$}{
        $B[i] \leftarrow B[i] / (n-i+1)$
    }
    $\rho_{max} \leftarrow \max_i B[i]$\tcp*[l]{maximum density in this iteration} 
    \ParFor{$v_i\ \KwIn\ V$}{
        $\ell(v_i) \leftarrow \ell(v_i) + A[i]$ \tcp*[l]{update loads}
    }
    \Return{$\ell, \rho_{max}$}
\end{algorithm}

\myparagraph{Parallelization of~\cref{alg:greedypp}}
We first describe how to parallelize all parts of \cref{alg:greedypp} except for Line 6.
Let $v_1, \ldots, v_n$ be an ordering of vertices for peeling. On the $i$'th iteration of the for-loop on Line 7, the induced subgraph of $G$ that we use is $S_i = G(v_i, \ldots, v_n)$. This holds for all $1 \leq i \leq n$. For any edge $e= (v_j, v_{k})$, 
$e$ will contribute to the density of $S_i$ if and only if $i \leq j$ and $i\leq k$.

Our implementation for computing the densities and updating the loads in parallel is shown in \cref{alg:peel_ordering}.
We first initialize an empty array $A$ of size $n$ (Line 1). Then, for each edge $e =(v_j, v_k)$, we add $1$ to $A[\min (j,k)]$ (Lines 2--3). Let $B$ be the suffix sum array of $A$ (Line 4). Then, $B[i]$ corresponds to the number of edges remaining in the graph after vertices $v_1, \ldots, v_{i-1}$ are peeled. 
To see why it is the case, let us consider the remaining subgraph after $v_1,\ldots, v_{i-1}$ are peeled.
Consider an edge $e = (v_j, v_k)$. Edge $e$ will appear in this subgraph if and only if both $v_j$ and $v_k$ are not yet peeled, i.e., $i \leq j$ and $i\leq k$. We add $1$ to $A[\min(j,k)]$ to account for the presence of $e$ in the subgraphs induced by $v_{i\leq \min(j,k)}, \ldots, v_n$. We then compute the densities in parallel and take the maximum density on Lines 5--7. We update the loads in parallel on Lines 8--9.
The work and the depth of this implementation are $O(n+m)$ and $O(\log n)$, respectively. 
As we described in \cref{sec:prelim}, \defn{ParFor}, \defn{SuffixSum}, and \defn{FindMax} all take linear work.
 \defn{SuffixSum} and \defn{FindMax} have $O(\log n)$ depth, and \defn{ParFor} have $O(1)$ depth, so our depth bound follows.

\myparagraph{\PGreedyPP}
In order to parallelize \GreedyPP, what is left for us is to parallelize the computation of the degeneracy ordering, which can be computed with a $k$-core decomposition algorithm in $O(m+n)$ expected work and $O(c_p\log n)$ depth, with high probability.
When we peel multiple vertices at the same time, the degeneracy ordering can be different from the order obtained sequentially.
The reason is that when a vertex is peeled in the sequential algorithm, it affects its neighbors' degrees immediately. However, in the parallel version, this effect is delayed until the end of the peeling step where multiple vertices may be peeled together.
We claim that this does not significantly affect the order.
Consider a pair of vertices $v_i$ and $v_{j}$. 
To make the proof in \cref{thm:sortingpp_work} go through, it suffices to show that $i<j$ implies $\ell(v_i) \leq \ell(v_j) + \Delta$.
We prove the contrapositive. 
Suppose $\ell(v_i) > \ell(v_j) + \Delta$. Because the number of neighbors of both $v_i$ and $v_j$ are bounded by $\Delta$,
it is the case that, even if all of $v_i$'s neighbors are peeled and none of $v_j$'s neighbors are peeled, $v_i$ will still be peeled after $v_j$, implying that $i>j$.
Therefore, for $i<j$ we have that $\ell(v_i) \leq \ell(v_j) +\Delta$.

\begin{thm}
\label{thm:parallel_greedypp_work}
For $T= \Theta \left( \frac{\Delta \log n}{\rho^* \eps^2} \right )$, our parallel algorithm \PGreedyPP~outputs a $(1+\eps)$-approximation to the densest subgraph problem. Moreover, each iteration takes $O(n+m)$ expected work and $O( c_p \log n)$ depth with high probability.
\end{thm}

\myparagraph{\PGreedySortingPP}
We 
replace the degeneracy order in \PGreedyPP\ with parallel \defn{IntegerSort} to obtain \PGreedySortingPP.
As discussed in~\cref{sec:prelim}, integer sorting takes 
$O(n \log \log n)$ expected work and $O(\log n)$ depth \whp{}, or  $O(n/\eps)$ work and $O(n^\eps)$ depth for any $0<\eps<1$. We proved earlier that sorting does not affect the approximation guarantee of the algorithm.
Therefore, we have the following theorem.

\begin{thm}
\label{thm:parallel_sortingpp_work}
For $T= \Theta\left( \frac{\Delta \log n}{\rho^* \eps^2} \right )$, our parallel \PGreedySortingPP\ outputs a $(1+\eps)$-approximation to the densest subgraph problem. Moreover, each iteration takes either $O(n\log\log n+m)$ expected work and $O(\log n)$ depth \whp{} or $O(n+m)$ work and $O(n^\eps)$ depth for any $0<\eps<1$.
\end{thm}

Note that for our algorithms, the number of steps $T$ depends 
on $\rho^*$, but we can choose $T$ 
based on $k_{max}$ instead, as it is within a factor of $2$ of $\rho^*$.
 
\myparagraph{Combining Refinement with Pruning}
Using the framework in \cref{alg:prune-and-iterate}, we can combine a pruning method with \emph{one iteration} of \GreedyPP, \GreedySortingPP, \PGreedyPP, and \PGreedySortingPP\ as the $Refine$ function.  
For example, we can combine
 approximate $k$-core for pruning with one iteration of \PGreedySortingPP, which would give
an algorithm with either $O(T(n\log\log n+m))$ expected work and $O(T \log n + \log^3n)$ depth or $O(T(n+m))$ work and $O(T\log n +n^\eps)$ depth for any $0<\eps<1$.

\myparagraph{Remark} While pruning gives speedups in practice, it does not 
improve the theoretical complexity of the algorithm, as there exists a graph where the core that we prune down to covers most of the original graph. However, as we observe in our experiments below, pruning results 
in massive improvements in runtime in practice as most real-world graphs exhibit a densest subgraph that is a small percentage of the input (in some graphs, the densest subgraph contains fewer than $1\%$ of the vertices
in the input).

\section{Experiments}
\label{sec:exp}

In this section, we implement and benchmark different instantiations of the Pruning-and-Refining framework 
on real-world datasets. 
We also compare our algorithms with existing algorithms. 
We demonstrate that our approach is practical and is scalable to the largest publicly-available graphs.
In addition, we also provide interesting statistical data on large-scale graphs, in particular, near-optimal densities of the larger graphs, previously not reported to such accuracy in literature. All of our code
is provided at 
\href{https://github.com/PattaraS/gbbs/tree/ALENEX}{this link}.\footnote{\href{https://github.com/PattaraS/gbbs/tree/ALENEX}{https://github.com/PattaraS/gbbs/tree/ALENEX}}

\myparagraph{Implementations}
We implement \GreedyPP, \GreedySortingPP, and their parallel instantiations as our refinement algorithms.
We consider two algorithms for pruning: pruning using exact $k$-cores, and pruning
using approximate $k$-cores. Both pruning algorithms are parallel and are modular across all of the refinement algorithms.
We use the exact $k$-core decomposition algorithm
from Dhulipala \etal~\cite{dhulipala2017julienne} and the
 approximate $k$-core decomposition algorithm from Liu \etal~\cite{LSYDS22} for our pruning step.
    The algorithm of Liu \etal~\cite{LSYDS22} allows us to specify the approximation ratio ($c$ in \cref{def:approx-k-core}). When $c$ is higher, the algorithm tends to be more parallelizable
    but will be less accurate. We run our experiments with $c=1.5$.
 We also combine 
 approximate $k$-core decomposition with exact $k$-core decomposition
 for further speedups. In particular, our combined pruning algorithm
 first uses approximate $k$-core decomposition to shrink the graph
 and then uses exact $k$-core decomposition for greater accuracy in 
 the refine step on the smaller graph. Such a procedure results in speedups for a peeling-based algorithm like \PGreedyPP since in 
 each iteration, the algorithm already performs much of the necessary work (with minimal modification) to find the $k$-core decomposition.
We name our algorithms as follows:
\begin{enumerate}[topsep=1pt,itemsep=0pt,parsep=0pt,leftmargin=15pt]
    \item \texttt{PaRGreedy++}: is \PGreedyPP\ combined with our Pruning-and-Refining framework with exact \kc pruning.
    \item \texttt{PaRSorting++}: is \PGreedySortingPP\ combined with our Pruning-and-Refining framework with exact \kc pruning.
    \item \texttt{\ApproxGP}: is \PGreedyPP\ combined with our Pruning-and-Refining framework with approximate \kc pruning followed by exact \kc pruning.
    \item \texttt{\ApproxGS}: is \PGreedySortingPP\ combined with our Pruning-and-Refining framework with approximate \kc pruning.
\end{enumerate}
We present experimental results for each of our methods and also
compare with existing implementations.

\myparagraph{Existing Algorithms} 
We compare with the state-of-the-art $(1+\eps)$-approximation algorithms from the sequential algorithms of Fang \etal~(\texttt{CoreExact}~\cite{FangYCLL19} and \texttt{CoreApp}~\cite{FangYCLL19}), the sequential algorithm of Boob \etal~(\GreedyPP~\cite{greedyPP}), the parallel algorithm of Harb \etal~(\FISTA~\cite{harb2022faster}, \texttt{Frank-Wolfe}~\cite{DCS17}, and \texttt{MWU}~\cite{BGM14}), and the sequential algorithm of Xu \etal~\cite{xu2023efficient}. We also compare with the state-of-the-art parallel $2$-approximation algorithms of Luo \etal~\cite{luo2023scalable} and Dhulipala \etal~\cite{dhulipala2018theoretically}. 

Fang et al.~\cite{FangYCLL19} and Xu \etal~\cite{xu2023efficient} implement variants of maximum flow algorithms to find the densest subgraph. Their \texttt{CoreExact}~\cite{FangYCLL19} and \ccexact~\cite{xu2023efficient} implementations are exactly the flow algorithm with binary search over the density. 
They perform a density lower-bound estimation using cores and approximate maximum flow, which we described in more detail in \cref{sec:pruning}. %
If a graph has many connected components, then a subgraph with maximum density lies exclusively in one component. Hence, they run the flow algorithm on each connected component separately. 
The densities of cores are then used to determine the lower bounds and upper bounds of the binary search needed for the flow computation. \ccexact~\cite{xu2023efficient} 
improves \texttt{CoreExact}~\cite{FangYCLL19} by using cores to shrink the input graph. 
Their approximation algorithm, \texttt{CoreApp}~\cite{FangYCLL19}
 is an algorithm that finds $core(G,k_{max})$ directly. Once the maximum core is found, they return the component with the highest density. This algorithm yields a $2$-approximation for the densest subgraph problem.
\ccg~\cite{xu2023efficient} is their implementation of \GreedyPP~\cite{greedyPP} that uses one iteration of a $k$-core decomposition algorithm to shrink the input graph.

Harb \etal~\cite{harb2022faster} propose a gradient descent based algorithm called \FISTA~\cite{harb2022faster}, where the number of iterations needed is $O(\sqrt{\Delta m}/\eps)$. 
They use accelerated proximal gradient descent, which is faster than the standard gradient descent approach~\cite{beck2009fast,nesterov1983method}. 
The algorithm runs in iterations, where each iteration can be made parallel.
The output in each iteration is a feasible solution to a linear program for the densest subgraph problem.
They then use \GreedyPP-inspired rounding, which they call fractional peeling, to round the linear program solution into an integral solution.

Finally, we benchmark against 
recent parallel algorithms, \julienne~\cite{dhulipala2018theoretically} and \pkmc~\cite{luo2023scalable}, for computing the exact \kc decomposition. 
Then, the maximum core gives a $2$-approximation of the densest subgraph. 
Although these algorithms achieve parallelism, their approximations are 
worse than the $(1+\eps)$-approximation algorithms, both theoretically
and empirically.

\myparagraph{Setup}
We use \texttt{c2-standard-60} Google Cloud instances (3.1 GHz
Intel Xeon Cascade Lake CPUs with a total of 30 cores with two-way
hyper-threading, and 236 GiB RAM) and \texttt{m1-megamem-96} Google
Cloud instances (2.0 GHz Intel Xeon Skylake CPUs with a total of
48 cores with two-way hyper-threading, and 1433.6 GB RAM). We
use hyper-threading in our parallel experiments by default. Our
programs are written in C++. We use parallel primitives from the \texttt{GBBS}~\cite{dhulipala2017julienne} and 
\texttt{Parlay}~\cite{blellochParlay} libraries.
The source code is compiled using g++ (version 10) with the -O3 flag. We
terminate experiments that take over 1 hour. %
We run each experiment for three times and take the average
for the runtime and accuracy analyses.
Using enough threads, with our framework, \GreedyPP, \GreedySortingPP, and their parallel versions finished within 1 hour for all of our experiments. However, all other $(1+\eps)$-approximation algorithms took longer than 1 hour on all of the large graphs (\texttt{clueweb}, \texttt{twitter}, \texttt{friendster}, and \texttt{hyperlink2012}), so we omit the entries for these datasets for these
algorithms.

\myparagraph{Datasets}
We run our experiments on various synthetic and real-world datasets. The real-world datasets are obtained from SNAP~\cite{leskovec2014snap} (\texttt{cahepth}, 
\texttt{ascaida}, \texttt{hepph}, \texttt{dblp}, \texttt{wiki}, \texttt{youtube}, \texttt{stackoverflow}, \texttt{livejournal}, \texttt{orkut}, \texttt{twitter}, and \texttt{friendster}), Network Repository~\cite{networkrepo} (\texttt{brain}), Lemur project at CMU~\cite{boldi2004webgraph} (\texttt{clueweb}), and WebDataCommons~\cite{Meusel2015TheGS} (\texttt{hyperlink2012}).
The \texttt{hyperlink2012} graph is the largest publicly-available real-world graph today.
\texttt{closecliques} is a synthetic dataset designed to be challenging for \GreedyPP~\cite{greedyPP,CQT22,harb2022faster}. 
We remove self-loop and zero-degree vertices from all graphs, and symmetrize  any directed graphs.
We run most of our experiments on \texttt{c2-standard-60} machines.
However, on the larger graphs (namely, \texttt{twitter}, \texttt{friendster}, \texttt{clueweb},
and \texttt{hyperlink2012}), we use \texttt{m1-megamem-96} machines as more memory is required. 
The sizes of our inputs and their maximum core values are included in~\cref{tab:graph_datasets}.

\myparagraph{Overview of Results}
We show the following experimental results in this section.

\begin{itemize}[topsep=0pt,itemsep=0pt,parsep=0pt,leftmargin=8pt]

\item Our pruning strategy is very efficient in practice as our pruned
graph contains $175\times$ fewer edges
on average and $3,596\times$ fewer vertices on average.
\item Our algorithms, similar to the state of the art, take only a few iterations to converge. \PARSortingPP takes more iterations, but it still converges to $<1.01$-approximation within $10$--$20$ iterations and each
iteration is significantly faster than all other algorithms.
\item Our algorithms are faster than existing algorithms by a large margin.
\item Our algorithms are highly parallelizable, achieving up to $22.37\times$ self-relative parallel speedup on a $30$-core machine with two-way hyperthreading. 
\item We measure empirical ``width'', which is a parameter that correlates to the number of iterations needed to converge. We observe that the empirical width is much smaller than the upper bound used to analyze the algorithm. This may lead to more fine-grained analyses of many MWU-inspired algorithms.

\end{itemize}

\subsection{Core-Based Pruning}
In this section, we present experimental results related to various different
pruning methods using exact and approximate $k$-core decomposition (and combinations thereof).

\myparagraph{Pruning with $core(G,\lceil \frac{k_{max}}{2}\rceil)$} 
We first study the benefit of performing pruning using the {\em exact} $k$-core computation.
The data for this experiment across all graphs is shown in~\cref{tab:graph_datasets}. 
For the real-world graphs, the cores contain between 2--282$\times$ fewer edges than the actual graphs ($48.3\times$ fewer on average), and between 4.5--14227$\times$ fewer vertices ($2420\times$ fewer on average).
The only exception is the \texttt{brain} dataset, where the core is half the size of the actual graph.
Even in this case, the number of vertices left in the core is around $25\%$ of the original graph.
For the synthetic dataset \texttt{closecliques}, the input is designed so that the maximum-core is identical to the original graph, so there is no benefit to pruning.
However, this situation is very unlikely to occur in real-world datasets. 

Due to the significant reduction in graph sizes in terms of both the number of vertices and number of edges,
using $core(G,\lceil \frac{k_{max}}{2}\rceil)$ is almost always preferable over using $G$, especially since computing all cores of $G$ (a linear-work algorithm, with reasonably high parallelism in practice~\cite{dhulipala2017julienne}) is inexpensive compared to the cost of running any of the refinement algorithms, which mostly require super-linear work.
To summarize, we find that pruning is nearly always beneficial and should be applied prior to refinement.

\myparagraph{Pruning with Highest Cores} 
As our algorithms progress, we perform additional pruning to shrink the graph even further. 
We report the sizes of the final graphs in~\cref{app:sec:tables} (\cref{tab:graph_datasets_max_den}). 
In many cases, the sizes of the final graph (after pruning to the highest cores)
are less than half of the sizes of their $\lceil\frac{k_{max}}{2}\rceil$-cores. 
Across all datasets,
we find that iterative pruning yields up to a $30.3\times$ reduction in the number of vertices over the $core(G, \ceil{k_{max}/2})$, and up to a $200\times$ reduction in the number of edges when comparing these quantities in $core(G,\lceil k_{max}/2 \rceil)$ and $core(G, \lceil 
{\tilde{\rho }} \rceil)$, where ${\tilde{\rho}}$ is the best density found by our algorithms. Thus, we see advantages in performing multiple rounds of refinement
in certain graphs.
Note that, there are cases when this additional pruning step is not helpful, e.g., in \texttt{dblp} and \texttt{clueweb}. In each of these cases, we notice that the best density found is very close to $\ceil{k_{max}/2}$, so there is no room for pruning opportunities.

\subsection{Number of Iterations Versus Density}

Next, we study the progress that different refinement algorithms make in our framework vs.\ other implementations toward finding the maximum density.
We perform this experiment on all variants of our algorithms: \PARGreedyPP{}, \PARSortingPP{}, \ApproxGP, and \ApproxGS; 
and all other benchmarks that use iterations:
\FISTA, \GreedyPP, \texttt{FrankWolfe}, \texttt{MWU}, and \ccg. We also 
include \pkmc as a baseline of comparison against a $2$-approximation algorithm. 
All algorithms were run for at least 20 iterations. 
The results are illustrated in \cref{fig:plot_den_iter}.
In fact, most algorithms converge very early with our algorithms \PARGreedyPP and \ApproxGP converging no later
than the fastest converging algorithms. In fact, on most graphs,
\GreedyPP, \PARGreedyPP, and \PARSortingPP\ took the fewest iterations to converge.
Two algorithms, \PARSortingPP\ and \texttt{MWU}, take more iterations in many graphs.
This matches with our understanding of the \emph{width} of MWU as discussed below. 
Furthermore, \ApproxGP and \ApproxGS exactly match the convergence
rates of \PARGreedyPP{} and \PARSortingPP{}, respectively; such is expected
as our use of approximate vs.\ exact pruning affects the runtime \emph{not} the 
accuracy. 
\ApproxGS needs more iteration to converge since we only use approximate \kc pruning.

\subsection{Approximation Ratio}

In \cref{tab:exp_approx_10}, 
we compare the densities returned from various algorithms at iteration $10$ with the best density currently known in the literature. Except for \texttt{brain}, \texttt{twitter}, \texttt{friendster}, \texttt{clueweb}, and \texttt{hyperlink2012}, the best known density is equal to the optimum. To compute the optimum density, we run a linear program solver on $core(G, \lceil \tilde \rho \rceil)$.
    
    \begin{table*}[h!]
    \begin{center}
    \footnotesize
    \begin{tabular}{l||c|r|r|r|r|r|r|r}
    
    \toprule
 Graph Dataset & $\tilde \rho$ & \texttt{FISTA} & \texttt{MWU} & \texttt{FrankWolfe} & \texttt{Greedy++} & \texttt{PaRGreedy++} & \texttt{PaRSorting++} & \texttt{\ApproxGS}\\
 \midrule
\texttt{hepph*}  & 265.969 & 1.00043 & 1.00011 & 1.00011 & 1 & 1 & 1.00031 & 1.00001 \\
\texttt{dblp*}  & 56.565 & 1 & 1 & 1 & 1 & 1 & 1 & 1 \\
\texttt{brain}  & 1057.458 & 1.00011 & 1.00005 & 1.00031 & 1 & 1 & 1.00026 & 1.0001 \\
\texttt{wiki*}  & 108.59 & 1 & 1 & 1.00723 & 1 & 1 & 1.00002 & 1 \\
\texttt{youtube*}  & 45.599 & 1.00023 & 1.00104 & 1.01522 & 1.00007 & 1 & 1.00079 & 1.00063 \\
\texttt{stackoverflow*}  & 181.587 & 1 & 1.00001 & 1.0031 & 1 & 1 & 1.00002 & 1.00001 \\
\texttt{livejournal*}  & 229.846 & 1.00113 & 1.00003 & 1.03671 & 1 & 1 & 1.00019 & 1.00006 \\
\texttt{orkut*}  & 227.874 & 1 & 1.00011 & 1.00123 & 1 & 1 & 1.00026 & 1.00026 \\
\texttt{twitter}  & 1643.301 & n/a & n/a & n/a & 1 & 1 & 1.00003 & 1.00006 \\
\texttt{friendster}  & 273.519 & n/a & n/a & n/a & n/a & 1 & 1 & 1.00006 \\
\texttt{clueweb}  & 2122.5 & n/a & n/a & n/a & n/a & 1 & 1 & 1 \\
\texttt{hyperlink2012}  & 6496.649 & n/a & n/a & n/a & n/a & 1 & 1 & 1.00002 \\
    \bottomrule
    \end{tabular}
    \end{center}
    \vspace{-5pt}
    \caption{Approximation Ratio at the 20th iteration for various algorithms. $\tilde{\rho}$ is the best densest subgraph known in the literature. Ratios are computed as the best currently known density ($\tilde \rho$) divided by the density produced by the respective algorithm. Results are indicated as n/a if the corresponding algorithms timeout at $1$ hour. Graphs indicated with an $*$ 
    have optimum computed densities.
    $\tilde{\rho}$ is rounded to 3 decimal places and approximation ratios are rounded to 5 decimal places.
    \ApproxGP is omitted since the ratios are identical to {\PARGreedyPP}.}
    \label{tab:exp_approx_20}
    \end{table*}

Except for \texttt{FrankWolfe}, all algorithms have approximation ratios less than $1.02$ after 10 iterations. Our \PARGreedyPP{} algorithm achieves the best 
approximation ratio after $10$ iterations for all four of the largest graphs,
\texttt{twitter}, \texttt{friendster}, \texttt{clueweb}, and \texttt{hyperlink2012}.
For the rest of the graphs, our algorithm achieves an approximation ratio
no worse than $1.0001\times$ the smallest approximation ratio.
We also include a table that compares densities after iteration 20 in~\cref{tab:exp_approx_20}.
After 20 iterations, most approximation ratios are less than $1.001$.
Our algorithm \PARGreedyPP{} obtains the best approximation for $8$ out of the 
$12$ tested graphs and obtains an approximation ratio no worse than 
$1.0000002\times$ the best for the remaining graphs.

\subsection{Empirical Widths}

As mentioned at the end of \cref{sec:algo}, in the multiplicative weight update framework, \defn{width} is a parameter that is correlated with
the number of rounds needed for a solution to converge.
In our context, the width $\omega$ corresponds to the {\em maximum increase of a load} of a single vertex in any iteration. See, e.g., \cite{arora2012multiplicative,CQT22} for more details on width and its analysis.
$\omega$ is lower bounded by
$\rho^*$ because 
when we peel the first vertex from the densest subgraph, its degree must be at least $\rho^*$.
The width is also upper bounded by
the maximum degree $\Delta$, since the increase of a load of a vertex is bounded by its degree, i.e., $\rho^* \leq \omega \leq \Delta$.
This upper bound is reflected in the
$T=O\left(\frac{\Delta \log n}{\rho^* \eps^2}\right)$ iterations needed for our algorithms in 
the worst case.
However, this bound does not reflect reality, as most of our
iterative algorithms usually converges in just a few iterations. 
The bound on the number of iterations could have been $T=O(\frac{\omega \log n}{\rho^* \eps^2})$. This can be significant if $\omega \ll \Delta$.
Here, we partially explain this phenomenon by measuring the width empirically.
In~\cref{tab:exp_width}, we report
the width across multiple datasets gathered from our experiments. 
We observe that the widths for running \PARGreedyPP\ are much closer to the best density found,
while the widths from \PARSortingPP\ are closer to $\Delta$.
If we see empirically that $\omega = O(\rho^*)$, then our algorithms should converge in $T = O\left(\frac{\log n}{ \eps^2}\right)$
iterations.
On the other hand,
if $\omega \gg \rho^*$, our algorithms should take more iterations to converge. 
This supports what we observed in our experiments, and also explains why it takes very few iterations, e.g., fewer than 10--20 iterations, for \PARGreedyPP\ to converge.

\begin{figure*}[htp!]
\centering
\includegraphics[width=\textwidth]{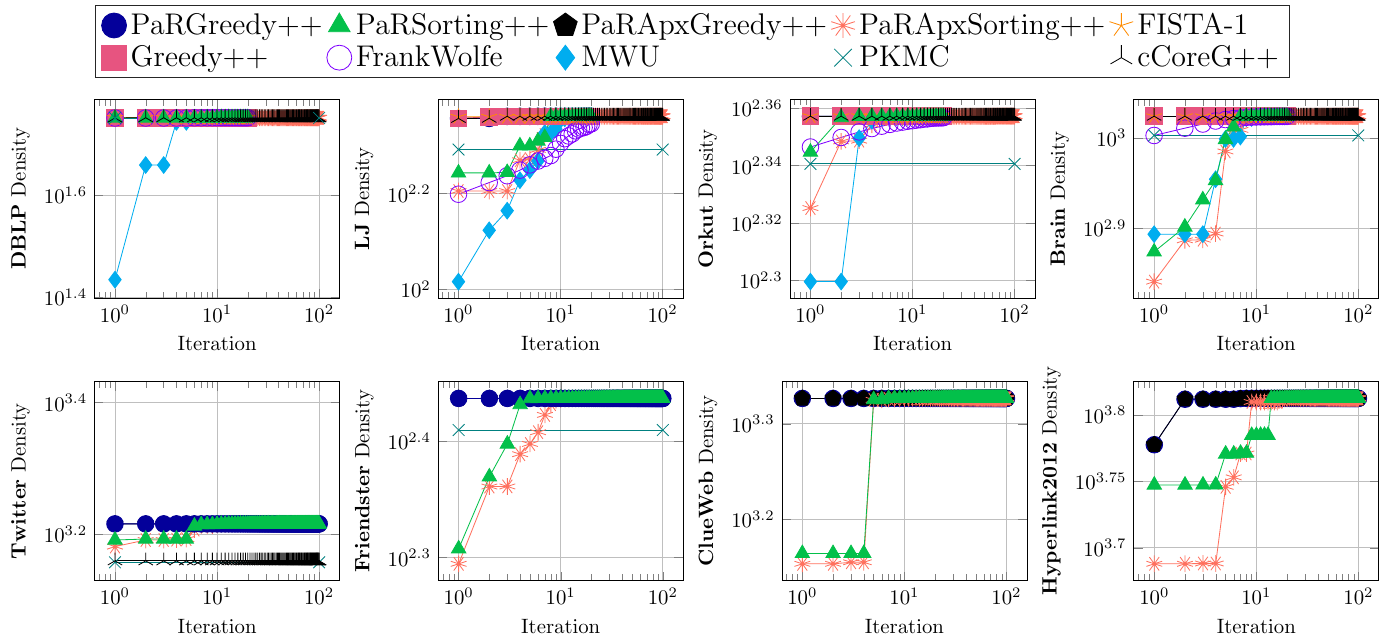}
\caption{Densities on different iterations for various algorithms. Only our algorithms can successfully process all of the large graphs (bottom row) within the $1$ hour limit. 
}\label{fig:plot_den_iter}
\end{figure*}

\begin{figure*}[htp!]
\centering
\includegraphics[width=\textwidth]{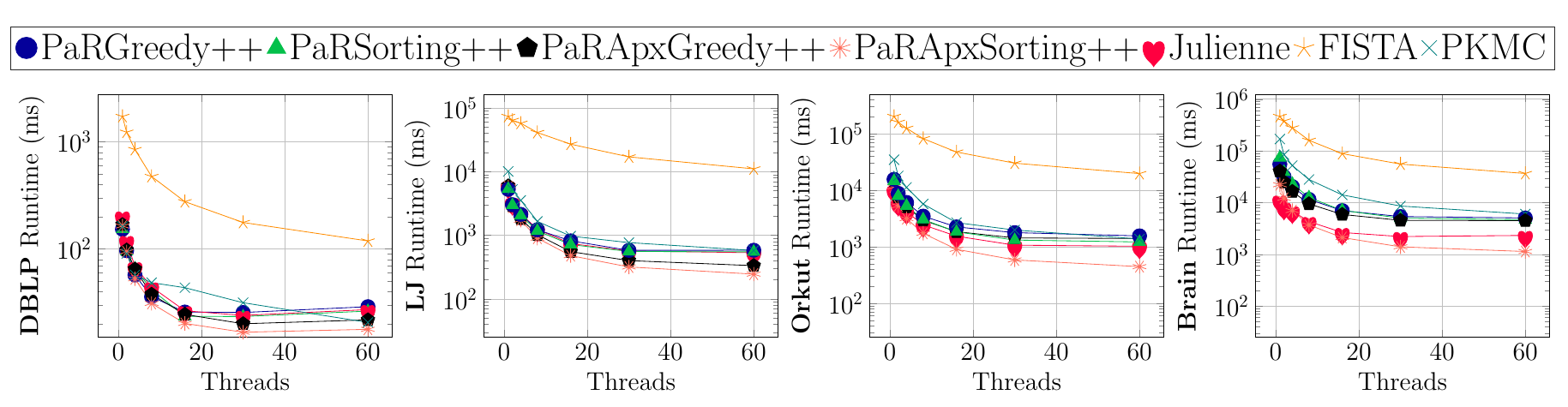}
\caption{Runtimes (ms) of \PARGreedyPP, \PARSortingPP, \ApproxGP, \ApproxGS, \julienne, \FISTA, and \pkmc versus the number of threads when running for $5$ iterations.
}
\label{fig:threads-5}
\end{figure*}

\begin{figure*}[htp!]
\begin{tikzpicture}
\begin{axis}[
	ylabel=Runtime (ms),
    ymode=log,
    bar width=1.5pt,
    ybar,
    ymax=50000000,
    clip=false,
    x tick label style={rotate=15},
    symbolic x coords={dblp, orkut, livejournal, stackoverflow, brain, wiki, youtube, hepph},
    xtick=data,
	enlargelimits=0.05,
    legend pos = north west,
    height=4cm,
    width=\textwidth,
    log origin=infty,
    legend columns=6,
   legend style={font=\scriptsize, at={(0.5, 1.5)},anchor=north}],
]
\addplot[fill=mydarkblue,postaction={pattern=grid}] %
	coordinates { %
        (brain,5814.427)
        (orkut,2030.765)
        (livejournal,730.386)
        (dblp,37.628)
        (stackoverflow,665.221)
        (wiki,128.375)
        (youtube,157.506)
        (hepph,594.574)
};
\addplot[fill=mygreen, postaction={pattern=dots}] %
	coordinates {
(brain,5661.566)
(orkut,1317.772)
(livejournal,579.395)
(dblp,31.173)
(stackoverflow,382.956)
(wiki,74.159)
(youtube,88.667)
(hepph,182.72)
};

\addplot[fill=ballblue,postaction={pattern=grid}] %
	coordinates { %
        (brain,6172.202)
        (orkut,2309.927)
        (livejournal,527.64)
        (dblp,31.85)
        (stackoverflow,1209.18)
        (wiki,140.924)
        (youtube,204.371)
        (hepph,1813.338)
};
\addplot[fill=bittersweet, postaction={pattern=dots}] %
	coordinates {
(brain,1172.47)
(dblp,18.436)
(hepph,31.912)
(livejournal,249.894)
(orkut,467.612)
(stackoverflow,166.007)
(wiki,30.346)
(youtube,41.472)
};

\addplot[fill=mydarkorange, postaction={pattern=horizontal lines}] %
	coordinates {
        (wiki,7540)
(dblp,492)
(stackoverflow,33730)
(livejournal,43986)
(orkut,76967)
(brain,138632)
(hepph,782)
(youtube,6386)
};
\addplot[fill=mydarkpurple, postaction={pattern = vertical lines}] %
	coordinates {
    (dblp,2637)
    (hepph,3527)
    (livejournal,182150)
    (orkut,466107)
    (stackoverflow,112867)
    (wiki,7390)
    (youtube,12160)
    (brain,1103504)
};
\addplot[fill=mypurple,postaction={pattern=dots}] %
	coordinates {
        (wiki,2612)
(dblp,1399)
(stackoverflow,41041)
(livejournal,69756)
(orkut,157189)
(brain,336626)
(hepph,1773)
(youtube,5509)
};
\addplot[fill=cyan,postaction={pattern=north west lines}] %
	coordinates {
        (wiki,4161) 
        (dblp,1334)
        (stackoverflow,20175)
        (livejournal,40709)
        (orkut,46364)
        (brain,44776)
        (hepph,272)
        (youtube,5088)
};

\addplot[fill=mydarkpink,postaction={pattern=dots}] %
    coordinates {
		(brain, 12503.7)
		(dblp, 224.396)
		(hepph, 329.461)
		(livejournal, 8435.12)
		(orkut, 13634.4)
		(stackoverflow, 4458.39)
		(wiki, 935.256)
		(youtube, 1135.52)
};

\addplot[fill=myyellow,postaction={pattern=north west lines}] %
    coordinates {
(brain, 943744)
(dblp, 319.661)
(hepph, 27682.8)
(livejournal, 18791.4)
(orkut, 244410)
(stackoverflow, 30408.2)
(wiki, 2804.46)
(youtube, 2739.82)
};

\addplot[fill=mydarkred,postaction={pattern=grid}] %
    coordinates {
		(brain, 6040.09)
		(dblp, 20.65)
		(hepph, 65.99)
		(livejournal, 577.748)
		(orkut, 1373.89)
		(stackoverflow, 365.107)
		(wiki, 153.929)
		(youtube, 110.132)
};

\legend{PaRGreedy++, PaRSorting++,PaRApxGreedy++,PaRApxSorting++, FISTA, Greedy++, FrankWolfe, MWU, cCoreG++, cCoreExact, PKMC}
\end{axis}
\end{tikzpicture}
\caption{Runtimes of different densest subgraph algorithms on our small graph inputs. The algorithms are run for $20$ iterations. Parallel algorithms use $60$ hyper-threads. %
}\label{fig:all-graphs}
\end{figure*}
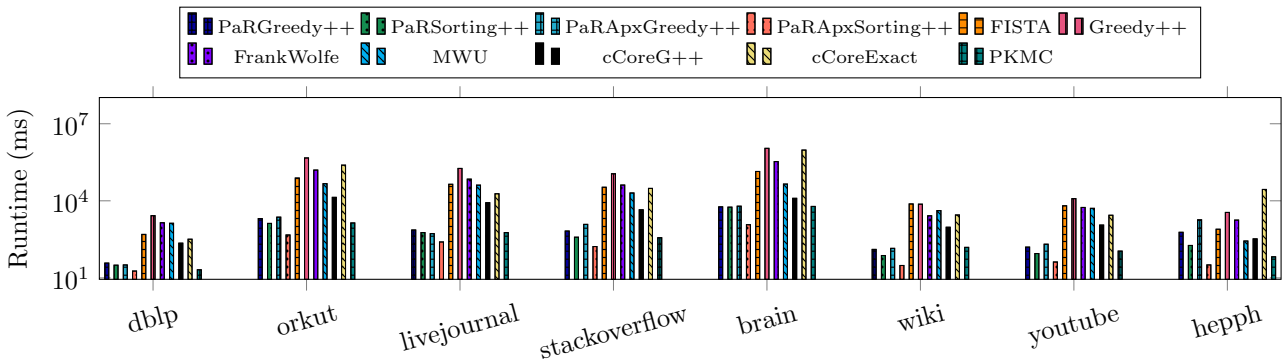

\subsection{Scalability}

Here, we show the scalability of our algorithms compared to the parallel version of FISTA and the parallel $2$-approximation algorithms, \julienne and \pkmc
in~\cref{fig:threads-5}, which shows the running time of the algorithms (in milliseconds)
versus the number of threads used by our algorithms and parallel \texttt{FISTA}, $\julienne$, and $\pkmc$ when each algorithm is run for $5$ iterations.
We show additional plots for $10$ and $20$ iterations in~\cref{fig:threads-10} and \cref{fig:threads-20}.
We see that our \texttt{PaRSorting++} algorithm
achieves greater self-relative
speedups than \texttt{FISTA} and \texttt{PaRGreedy++}. 
Specifically, \texttt{PaRSorting++} achieves up to a $10.6\times$ self-relative speedup (on \texttt{livejournal}),
while \texttt{FISTA} achieves up to a $14\times$ self-relative speedup (on \texttt{dblp}) and \texttt{PaRGreedy++} achieves up to a 
$5.51\times$ self-relative speedup (on \texttt{orkut}). Furthermore, both of our algorithms take shorter time than parallel \texttt{FISTA}
regardless of the number of threads. Our \ApproxGP and \ApproxGS achieve
the greatest self-relative speedup. Specifically, On \texttt{livejournal}, \ApproxGP and \ApproxGS achieves 
up to a $17.6\times$ and a $20.51\times$ self-relative speedup, respectively.
\ApproxGS achieves greater
self-relative speedup than the rest of the implementations for $8$ of the 
$12$ tested graphs. Although \pkmc achieves greater self-relative
speedups on \texttt{dblp}, \texttt{hepph}, \texttt{stackoverflow}, and \texttt{friendster}, they obtain worse approximations guarantees since they
only guarantee a $2$-approximation on the density. As we discussed 
previously, it is much easier to obtain greater parallelism 
when the approximation guarantee is relaxed to a $2$-approximation (we see in~\cref{fig:plot_den_iter} that \pkmc obtains 
noticeably worse approximations on most graphs).

\subsection{Comparing the Total Running Time}

We first compare \PARSortingPP\ with the algorithms given in~\cite{FangYCLL19} (see~\cref{table:vldb_compare}).
We ran experiments on two of their algorithms, namely, \texttt{CoreExact} and \texttt{CoreApp}.
\texttt{CoreExact} took too long to run on most datasets.
\texttt{CoreApp} is faster than our implementation on some graphs, however, since it uses $core(G,k_{max})$, this algorithm gives a $2$-approximation and is less accurate than our algorithms. 

We include plots that compare the total runtime of our algorithms (\PARGreedyPP, \PARSortingPP, \ApproxGP, \ApproxGS) with \GreedyPP~\cite{greedyPP}, \FISTA, \texttt{FrankWolfe}, \pkmc, \ccg, \ccexact, and \texttt{MWU}~\cite{harb2022faster} in~\cref{fig:all-graphs}. %
In short, when measuring the quality of our solutions in running time, our algorithms outperform all existing algorithms by significant margins, and is up to $25.9\times$ faster than the fastest $(1+\eps)$-approximation algorithm for each graph. 
On many of the graphs that we tested, we achieve a $2\times$ improvement in runtime when using 
approximate \kc compared to when we use exact.

\myparagraph{Large Graph Runtime and Accuracy Results} Even when using multi-threading, 
our algorithms are the only algorithms to finish processing the large graphs (\texttt{twitter}, 
\texttt{friendster}, \texttt{clueweb}, and \texttt{hyperlink2012}) within $1$ hour. 
Specifically, for $20$ iterations and $60$ threads, our fastest algorithms \PARGreedyPP{} 
and \PARSortingPP{} require 
$10.44$, $15.35$, $112.64$, and $352.65$ seconds, and $8.41$, $10.54$, $83.91$, and $270.39$
seconds on \texttt{twitter}, \texttt{friendster}, \texttt{clueweb}, and \texttt{hyperlink2012}, respectively.
When using approximate \kc, \ApproxGP{} and \ApproxGS{} require
$9.77$, $15.15$, $120.97$, and $375.71$ seconds, and $8.27$, $8.62$, $85.29$, and $287.27$ 
seconds on \texttt{twitter}, \texttt{friendster}, \texttt{clueweb}, and \texttt{hyperlink2012}, respectively.
Moreover, we obtain the best densities known in literature for these graphs (shown in~\cref{tab:exp_approx_20}). 
For massive real-world graphs, our experiments show that using approximate \kc 
does not yield much benefit. 
This is because larger graphs lead to more parallelism in the exact $k$-core decomposition algorithm,
so the exact $k$-core  algorithm exhibits similar parallelism 
to the approximate $k$-core algorithm on a $30$-core machine. 

\myparagraph{Initialization time} 
We measure the initialization time of our algorithms (i.e., the time to perform the \kc decomposition step), and report the results in~\cref{tab:graph_inittime_thread_algos}.
The finding is that a significant portion of total runtime is on this initialization step for all of the graphs. 
This makes sense as the initial graph tends to be much larger than the graph we obtain after pruning.
To illustrate this point, for \texttt{brain}, when running on $60$ threads, the initialization step takes at least $2655$ ms for \PARGreedyPP, \PARSortingPP, and \ApproxGP, and $1074$ ms for \ApproxGS. The average time spent on each iteration for these algorithms are $138$, $3$, $148$, and $2.3$ ms, respectively.
\ApproxGS is still faster than all other algorithms despite running for $300$ more iterations.

\section{Conclusion}\label{sec:conclusion}
We introduced a framework that combines pruning and refinement for solving the approximate densest subgraph problem. 
We designed new parallel variants of the 
sequential \GreedyPP\ algorithm, and achieved state-of-the-art performance by plugging them into our framework. 
We showed that our algorithms can scale to the large \texttt{hyperlink2012} and \texttt{clueweb} graphs and obtain near-optimal approximations
of their densest subgraphs for the first time in the literature.

\section*{Acknowledgements}
This research was supported by DOE
Early Career Award \#DESC0018947, NSF Awards
\#CCF-1845763 and \#CCF-2103483, Google Faculty Research Award, Google Research Scholar Award, cloud computing credits from Google-MIT, and FinTech@CSAIL Initiative.

\clearpage
\bibliographystyle{alpha}
\bibliography{ref}

\clearpage
\onecolumn
\appendix
\section{Additional Tables from Section~\ref{sec:exp}}
\label{app:sec:tables}
        \begin{table*}[h!]
    \begin{center}
    \footnotesize
    \begin{tabular}{l||rr|r|rr|cc}
    
    \toprule
    & \multicolumn{2}{c|}{Original Graph}&& \multicolumn{2}{c|}{$core(G,\lceil \frac{k_{max}}{2} \rceil)$}\\

    Graph Dataset & Num. Vertices & Num. Edges & $k_{max}$ & Num. Vertices & Num. Edges & Vertex Ratio & Edge Ratio\\
    
    \midrule
        \texttt{closecliques}  & 3,230 & 95,400 & 59 & 3,230 & 95,400 & 1.000 & 1.000\\
        \texttt{cahepth}  & 9,877 & 25,973 & 31 & 96 & 2,306 & 0.0097 & 0.088\\
        \texttt{ascaida}  & 26,475 & 106,762 & 22 & 208 & 6,244 & 0.007 & 0.048\\
        \texttt{hepph} & 28,094 & 3,148,447 & 410 & 6,304 & 1,562,818 & 0.224 & 0.494 \\
        \texttt{dblp}  & 317,080 & 1,049,866 & 101 & 280 & 13,609 & 0.001 & 0.010\\
        \texttt{brain}  & 784,262 & 267,844,669 & 1,200 & 187,494 & 137,354,946 & 0.239 & 0.512\\
        \texttt{wiki}  & 1,094,018 & 2,787,967 & 124 & 3,807 & 344,553 & 0.034 & 0.090\\
        \texttt{youtube}  & 1,138,499 & 2,990,443 & 51 & 12,836 & 439,678 & 0.011 & 0.110\\
        \texttt{stackoverflow}  & 2,584,164 & 28,183,518 & 163 & 41,651 & 5,709,796 & 0.016 & 0.187\\
        \texttt{livejournal}  & 4,846,609 & 42,851,237 & 329 & 6,090 & 1,054,941 & 0.001 & 0.022\\
        \texttt{orkut}  & 3,072,441 & 117,185,083 & 253 & 71,507 & 13,469,722 & 0.023 & 0.113\\
        \texttt{twitter}  & 41,652,230 & 1,202,513,046 & 2,484 & 24,480 & 36,136,023 & 0.001 & 0.029\\
        \texttt{friendster}  & 65,608,366 & 1,806,067,135 & 304 & 1,474,236 & 271,902,207 & 0.022 & 0.146\\
        \texttt{clueweb}  & 978,408,098 & 37,372,179,311 & 4,244 & 91,874 & 132,549,663 & 9.39e-05 & 0.003\\
        \texttt{hyperlink2012}  & 3,563,602,789 & 112,920,331,616 & 10,565 & 250,477 & 1,046,929,322 & 7.02e-05 & 0.009\\
        
    \bottomrule
    \end{tabular}
    \end{center}
    \caption{Graph sizes, their maximum core values ($k_{max}$), their $\lceil k_{max}/2\rceil$-core sizes, and their vertex (edge) ratios, which is the number of vertices (edges) in $core(G,\lceil \frac{k_{max}}{2} \rceil)$ divided by the number of vertices (edges) in $G$.}
    
    \label{tab:graph_datasets}
    \end{table*}

    \begin{table*}[h!]
    \begin{center}
    \footnotesize
    \begin{tabular}{l||r|rr|r|rr|cc}
    
    \toprule
    & &\multicolumn{2}{c|}{$core(G,\lceil \frac{k_{max}}{2} \rceil)$}&& \multicolumn{2}{c|}{$core(G,\lceil \tilde \rho \rceil )$}\\
    
Graph Dataset & $k_{max}$ & Num. Vertices & Num. Edges & $\tilde \rho$ & Num. Vertices & Num. Edges & Vertex Ratio & Edge Ratio \\
\midrule
\texttt{closecliques} & 59 & 3,230 & 95,400 & 29.55665 & 3,230 & 95,400 & 1.000 & 1.000\\
\texttt{cahepth} & 31 & 96 & 2,306 & 15.5 & 96 & 2,306 & 1.000 & 1.000\\
\texttt{ascaida} & 22 & 208 & 6,244 & 17.53409 & 90 & 1,578 & 0.432 & 0.259\\
\texttt{hepph} & 410 & 6304 & 1,562,818 & 265.969 & 3,786 & 988,734 & 0.600 & 0.633 \\ 
\texttt{dblp} & 101 & 280 & 13,609 & 56.56522 & 280 & 13,609 & 1.000 & 1.000\\
\texttt{brain} & 1,200 & 187,494 & 137,354,946 & 1,057.458 & 8,993 & 9,509,717 & 0.048 & 0.069\\
\texttt{wiki} & 124 & 3,807 & 344,553 & 108.5877 & 1,379 & 149,739 & 0.362 & 0.434\\
\texttt{youtube} & 51 & 12,836 & 439,678 & 45.59877 & 2,269 & 103,342 & 0.176 & 0.233\\
\texttt{stackoverflow} & 163 & 41,651 & 5,709,796 & 181.5867 & 5,877 & 1,067,149 & 0.141 & 0.187\\
\texttt{livejournal} & 329 & 6,090 & 1,054,941 & 229.8459 & 3,393 & 610,864 & 0.557 & 0.579\\
\texttt{orkut} & 253 & 71,507 & 13,469,722 & 227.874 & 26,670 & 6,077,055 & 0.372 & 0.451\\
\texttt{twitter} & 2,484 & 24,480 & 36,136,023 & 1643.301 & 11,619 & 17,996,107 & 0.474 & 0.498\\
\texttt{friendster} & 304 & 1,474,236 & 271,902,207 & 273.5187 & 49,370 & 1,353,583 & 0.033 & 0.005\\
\texttt{clueweb} & 4,244 & 91,874 & 132,549,663 & 2122.5 & 91,874 & 132,549,663 & 1.000 & 1.000\\
\texttt{hyperlink2012} & 10,565 & 250,477 & 1,046,929,322 & 6,496.649 & 154,410 & 754,065,464 & 0.616 & 0.720\\

    \bottomrule
    \end{tabular}
    \end{center}
    \caption{Comparison between $core(G,\lceil k_{max}/2 \rceil)$ and $core(G, \lceil 
{\tilde{\rho }} \rceil)$, where ${\tilde{\rho}}$ is best density found by our algorithms. The vertex (edge) ratio is the ratio of vertices (edges) in $core(G, \lceil \tilde{\rho} \rceil)$ to the number of vertices (edges) in $core(G, \lceil k_{max}/2\rceil)$.}
    \label{tab:graph_datasets_max_den}
    \end{table*}

        \begin{table*}[h!]
    \begin{center}
    \footnotesize
    \begin{tabular}{l||r|r|r|r|r|r|r|r}
    
    \toprule
 Graph Dataset & $\tilde \rho$ & \texttt{FISTA} & \texttt{MWU} & \texttt{FrankWolfe} & \texttt{Greedy++} & \texttt{PaRGreedy++} & \texttt{PaRSorting++} & \texttt{\ApproxGS}\\
 \midrule
\texttt{hepph*}  & 265.969   & 1.01471 & 1.00024 & 1.00033 & 1.00007 & 1.00011 & 1.00271 & 1.00018 \\
\texttt{dblp*}  & 56.565   & 1 & 1 & 1 & 1 & 1 & 1 & 1 \\
\texttt{brain}  & 1057.458   & 1.00011 & 1.00161 & 1.00185 & 1 & 1 & 1.0022 & 1.00186 \\
\texttt{wiki*}  & 108.59  & 1 & 1.00001 & 1.02329 & 1 & 1 & 1.00009 & 1.00001 \\
\texttt{youtube*}  & 45.599   & 1.00023 & 1.00379 & 1.04566 & 1.0002 & 1.00019 & 1.00171 & 1.00151 \\
\texttt{stackoverflow*}  & 181.587   & 1 & 1.00006 & 1.01409 & 1 & 1 & 1.00009 & 1.00011 \\
\texttt{livejournal*}  & 229.846   & 1.00242 & 1.01705 & 1.13958 & 1 & 1 & 1.023 & 1.00038 \\
\texttt{orkut*}  & 227.874   & 1 & 1.00035 & 1.00359 & 1 & 1 & 1.00033 & 1.00113 \\
\texttt{twitter}  & 1643.301   & n/a & n/a & n/a & 1 & 1 & 1.00268 & 1.00216 \\
\texttt{friendster}  & 273.518   & n/a & n/a & n/a & n/a & 1 & 1 & 1.00191 \\
\texttt{clueweb}  & 2122.5   & n/a & n/a & n/a & n/a & 1 & 1 & 1.00071 \\
\texttt{hyperlink2012}  & 6496.649   & n/a & n/a & n/a & n/a & 1 & 1.0667 & 1.00631 \\

    \bottomrule
    \end{tabular}
    \end{center}
    \caption{Approximation ratio at the 10th iteration for various algorithms. $\tilde{\rho}$ is the best densest subgraph known in the literature. Ratios are computed as the best currently known density ($\tilde \rho$) divided by the density produced by the respective algorithm. Results are indicated as n/a if the corresponding algorithms timeout at $1$ hour. Graphs indicated with an $*$ 
    have optimum computed densities.
    $\tilde{\rho}$ is rounded to 3 decimal places and approximation ratios are rounded to 5 decimal places.
    \ApproxGP is omitted since their approximation ratios are identical to those of {\PARGreedyPP}.}
    \label{tab:exp_approx_10}
    \end{table*}

        \begin{table*}[h!]
    \begin{center}
    \footnotesize
    \begin{tabular}{l||r|r|rrr|rr|rr}
    
    \toprule
    &&&\multicolumn{3}{c|}{Max. Degree $\Delta$} & \multicolumn{2}{c|}{\PARGreedyPP} & \multicolumn{2}{c}{\PARSortingPP}\\
    Graph Dataset & Num. Vertices & $\tilde \rho$ & $G$ & $\lceil \frac{k_{max}}{2}\rceil$-core & $\lceil \tilde \rho \rceil$-core  & No Pruning & Pruning & No Pruning & Pruning \\
    \midrule
\texttt{closecliques} & 3,230 & 29.55665 & 2,000 & 2,000 & 2,000 & 59 & 59 & 2,000 & 2,000\\
\texttt{cahepth} & 9,877 & 15.5 & 65 & 31 & 31 & 31 & 31 & 65 & 31\\
\texttt{ascaida} & 26,475 & 17.53409 & 2,628 & 146 & 76 & 44 & 45 & 2,628 & 56\\
\texttt{hepph} & 28094 & 265.969 & 4,909 & 3,259 & 2,589 & 683 & 633 & 4,909 & 2,543\\
\texttt{dblp} & 317,080 & 56.56522 & 343 & 114 & 114 & 114 & 114 & 343 & 114\\
\texttt{brain} & 784,262 & 1,057.458 & 21,743 & 16,151 & 7,686 & 2,545 & 2,676 & 21,743 & 13,523\\
\texttt{wiki} & 1,094,018 & 108.5877 & 141,951 & 2,288 & 1,025 & 312 & 338 & 141,951 & 1,023\\
\texttt{youtube} & 1,138,499 & 45.59877 & 28,754 & 4,064 & 1,108 & 137 & 139 & 28,754 & 1,392\\
\texttt{stackoverflow} & 2,584,164 & 181.5867 & 44,065 & 14,613 & 3,783 & 640 & 619 & 44,065 & 3,787\\
\texttt{livejournal} & 4,846,609 & 229.8459 & 20,333 & 1,010 & 629 & 546 & 535 & 20,333 & 1,010\\
\texttt{orkut} & 3,072,441 & 227.874 & 33,313 & 13,162 & 7,186 & 755 & 779 & 33,313 & 7,447\\
\texttt{twitter} & 41,652,230 & 1,643.301 & 2,997,487 & 19,549 & 10,705 & 3,859 & 3,737 & 2,997,487 & 10,286\\
\texttt{friendster} & 65,608,366 & 273.5187 & 5,214 & 2,952 & 2,190 & 1,092 & 1,018 & 5,214 & 2,431\\
\texttt{clueweb} & 978,408,098 & 2,122.5 & 75,611,696 & 7,707 & 7,707 & 6,661 & 7,065 & 75,611,696 & 7,707\\
\texttt{hyperlink2012} & 3,563,602,789 & 6,496.649 & 95,041,164 & 67,920 & 3,740 & n/a & 17,287 & 95,041,163 & 67,792\\

    \bottomrule
    \end{tabular}
    \end{center}
    \caption{Empirical widths in our experiments.}
    \label{tab:exp_width}
    \end{table*}

        \begin{table*}[h!]
    \begin{center}
    \footnotesize
    \begin{tabular}{l||r|rr|rr|r}
\toprule
 &   &\multicolumn{2}{c|}{ \texttt{CoreApp}}   & \multicolumn{2}{c|}{\texttt{PaRSorting++}}   & \texttt{CoreExact}\\
Graph Dataset & $\tilde \rho$ & Density & time(ms) & Density & time(ms) & time(ms)\\
\midrule
\texttt{asCaida} & 17.53 & 16.72 & 4704 & 17.53 & 13 & 4704\\
\texttt{caHepTh} & 15.5 & 15.5 & 417 & 15.5 & 6 & 417\\
\texttt{brain} & 1057.45 & 1006.94 & 23573 & 1057.45 & 76959 & n/a\\
\texttt{dblp} & 56.57 & 56.5 & 65 & 56.57 & 154 & 247018\\
\texttt{hepph} & 265.97 & 205 & 61 & 265.97 & 1351 & n/a\\
\texttt{lj} & 229.85 & 195.86 & 2295 & 229.85 & 5344 & n/a\\
\texttt{orkut} & 227.87 & 219.32 & 17757 & 227.87 & 14740 & n/a\\
\texttt{stackoverflow} & 181.59 & 173.23 & 3599 & 181.59 & 4462 & n/a\\
\texttt{wiki} & 108.59 & 102.12 & 408 & 108.59 & 487 & n/a\\
\texttt{youtube} & 45.60 & 43.03 & 273 & 45.60 & 642 & n/a\\
    \bottomrule
    \end{tabular}
    \end{center}
    \caption{Comparison between \PARSortingPP\ on one thread and algorithms from~\cite{FangYCLL19}.}
    \label{table:vldb_compare}
    \end{table*}

\begin{table*}[h]
    \begin{center}
    \footnotesize
    \begin{tabular}{l||rr|rr|rr}
    \toprule
    Graph & \GreedyPP & \GreedySortingPP &\PARGreedyPP & \PARSortingPP & \ApproxGP & \ApproxGS \\
    \midrule
   \texttt{brain} & 1057.4578 & 1057.2957 & 1057.4578 & 1057.4577 & 1057.4578 & 1057.4570 \\
   \texttt{dblp} & 56.5652 & 56.5652 & 56.5652 & 56.5652 & 56.5652 & 56.5652 \\
   \texttt{hepph} & 265.9685 & 265.9559 & 265.9592 & 265.9294 & 265.9592 & 265.9689 \\
   \texttt{lj} & 229.8459 & 216.5275 & 229.8395 & 229.8318 & 229.8395 & 229.8459 \\
   \texttt{orkut} & 227.8740 & 227.5191 & 227.8682 & 227.8612 & 227.8682 & 227.8734 \\
   \texttt{stackoverflow} & 181.5867 & 181.4131 & 181.5844 & 181.5858 & 181.5844 & 181.5867 \\
   \texttt{wiki} & 108.5900 & 108.2407 & 108.5894 & 108.5898 & 108.5894 & 108.5900 \\
   \texttt{youtube} & 45.5988 & 45.1702 & 45.5722 & 45.5831 & 45.5722 & 45.5988 \\
   \texttt{twitter} & n/a & n/a & 1643.301712 & 1638.910735 & 1643.301712 & 1639.76604 \\
   \texttt{friendster} & n/a & n/a &  273.5180 &  273.5180 &  273.5180 &  273.5186 \\
   \texttt{clueweb} & n/a & n/a & 2122.499529 & 2122.499529 & 2122.499529 & 2122.499529 \\
   \texttt{hyperlink2012} & n/a & n/a & 6496.6493 & 6496.5325 & 6496.6493 & 6496.6493 \\
    \bottomrule

    \end{tabular}
    \end{center}
    \caption{Densities found after $20$ iterations ($100$ iterations for \ApproxGS)}
    \label{tab:graph_density_algos}
    \end{table*}

\begin{table*}[h]
    \begin{center}
    \footnotesize
    \begin{tabular}{ll||rr|rr|rr}
    \toprule
    Graph & Thread & \GreedyPP & \GreedySortingPP &\PARGreedyPP & \PARSortingPP & \ApproxGP & \ApproxGS \\
    \midrule
   \texttt{brain} & 1 & 220952.0 & 28582.1 & 26808.5 & 20144.3 & 43326.7 & 25838.2 \\
   \texttt{brain} & 2 & 177920.0 & 14597.6 & 19744.6 & 15794.5 & 27776.9 & 13991.9 \\
   \texttt{brain} & 4 & 166769.0 & 8797.73 & 14976.7 & 9459.33 & 18578.5 & 8113.02 \\
   \texttt{brain} & 8 & 109949.0 & 4787.92 & 9297.5 & 5608.43 & 11002.9 & 4263.11 \\
   \texttt{brain} & 16 & 85629.9 & 2487.92 & 6244.45 & 3489.77 & 7100.09 & 2428.17 \\
   \texttt{brain} & 30 & 77753.9 & 1404.68 & 5196.04 & 2690.38 & 5629.55 & 1603.39 \\
   \texttt{brain} & 60 & 84611.1 & 1098.11 & 5428.55 & 2714.91 & 5827.92 & 1306.19 \\
    \midrule
   \texttt{dblp} & 1 & 3533.22 & 746.742 & 143.232 & 133.875 & 175.727 & 172.001 \\
   \texttt{dblp} & 2 & 2458.81 & 406.275 & 95.881 & 86.882 & 105.115 & 101.425 \\
   \texttt{dblp} & 4 & 2032.91 & 279.44 & 73.95 & 52.097 & 72.714 & 59.442 \\
   \texttt{dblp} & 8 & 1470.16 & 141.829 & 41.179 & 32.945 & 45.005 & 36.64 \\
   \texttt{dblp} & 16 & 1208.5 & 74.083 & 29.459 & 24.274 & 31.37 & 22.931 \\
   \texttt{dblp} & 30 & 1113.5 & 51.771 & 30.46 & 22.847 & 27.948 & 19.409 \\
   \texttt{dblp} & 60 & 1259.32 & 45.628 & 36.815 & 25.436 & 31.219 & 21.418 \\
    \midrule
   \texttt{hepph} & 1 & 2611.14 & 136.274 & 692.539 & 234.824 & 845.31 & 477.387 \\
   \texttt{hepph} & 2 & 2663.3 & 77.069 & 689.243 & 181.719 & 762.716 & 278.305 \\
   \texttt{hepph} & 4 & 2135.64 & 53.577 & 612.173 & 119.875 & 531.692 & 157.652 \\
   \texttt{hepph} & 8 & 2203.48 & 28.795 & 475.481 & 90.68 & 411.762 & 101.529 \\
   \texttt{hepph} & 16 & 1987.48 & 20.69 & 390.74 & 77.119 & 383.619 & 60.612 \\
   \texttt{hepph} & 30 & 2044.24 & 16.326 & 438.746 & 86.953 & 411.679 & 57.718 \\
   \texttt{hepph} & 60 & 2517.87 & 17.031 & 488.445 & 103.12 & 513.2 & 57.273 \\
    \midrule
   \texttt{lj} & 1 & 108607.0 & 20304.8 & 4936.7 & 4683.37 & 6127.28 & 5890.9 \\
   \texttt{lj} & 2 & 65743.0 & 10066.3 & 3039.3 & 2745.38 & 3091.35 & 2852.85 \\
   \texttt{lj} & 4 & 50161.5 & 6302.07 & 2222.64 & 1979.14 & 2010.01 & 1796.11 \\
   \texttt{lj} & 8 & 31722.8 & 3369.24 & 1376.23 & 1134.88 & 1162.66 & 950.565 \\
   \texttt{lj} & 16 & 21764.1 & 1711.91 & 865.884 & 691.665 & 677.034 & 513.773 \\
   \texttt{lj} & 30 & 17975.0 & 1037.72 & 701.209 & 531.227 & 517.277 & 345.649 \\
   \texttt{lj} & 60 & 18962.6 & 668.648 & 698.862 & 495.418 & 482.402 & 271.28 \\
    \midrule
   \texttt{orkut} & 1 & 191907.0 & 24002.7 & 13735.7 & 9988.65 & 16773.2 & 13111.2 \\
   \texttt{orkut} & 2 & 121504.0 & 12029.7 & 8935.86 & 5941.39 & 9523.34 & 6178.15 \\
   \texttt{orkut} & 4 & 98716.8 & 7196.5 & 6875.54 & 4285.67 & 6595.43 & 3751.9 \\
   \texttt{orkut} & 8 & 61121.6 & 3823.59 & 4342.05 & 2523.46 & 4021.5 & 2075.75 \\
   \texttt{orkut} & 16 & 42554.3 & 1983.26 & 3004.36 & 1685.85 & 2669.04 & 1123.17 \\
   \texttt{orkut} & 30 & 34377.5 & 1195.75 & 2373.1 & 1133.72 & 2138.11 & 743.157 \\
   \texttt{orkut} & 60 & 35795.4 & 862.46 & 2457.03 & 1050.73 & 2180.69 & 586.859 \\
    \midrule
   \texttt{stackoverflow} & 1 & 57519.5 & 10986.7 & 3830.47 & 2889.86 & 4853.17 & 4030.38 \\
   \texttt{stackoverflow} & 2 & 35733.6 & 5544.66 & 2657.3 & 1716.76 & 3007.77 & 2077.23 \\
   \texttt{stackoverflow} & 4 & 28144.2 & 3385.08 & 2120.63 & 1219.33 & 2224.51 & 1265.24 \\
   \texttt{stackoverflow} & 8 & 18271.6 & 1789.24 & 1502.94 & 704.916 & 1510.7 & 702.252 \\
   \texttt{stackoverflow} & 16 & 12989.2 & 953.728 & 1150.8 & 414.358 & 1157.95 & 387.525 \\
   \texttt{stackoverflow} & 30 & 10982.0 & 568.138 & 1080.21 & 312.7 & 1088.32 & 294.483 \\
   \texttt{stackoverflow} & 60 & 11533.3 & 391.7 & 1212.31 & 312.277 & 1187.38 & 231.058 \\
    \midrule
   \texttt{wiki} & 1 & 8276.28 & 2573.86 & 439.37 & 349.467 & 522.626 & 441.785 \\
   \texttt{wiki} & 2 & 5124.66 & 1310.03 & 316.066 & 219.355 & 355.683 & 254.214 \\
   \texttt{wiki} & 4 & 3491.19 & 789.241 & 201.594 & 148.295 & 247.615 & 159.569 \\
   \texttt{wiki} & 8 & 2264.19 & 425.664 & 151.524 & 88.289 & 151.98 & 80.512 \\
   \texttt{wiki} & 16 & 1544.02 & 219.34 & 117.327 & 56.037 & 127.026 & 52.78 \\
   \texttt{wiki} & 30 & 1368.65 & 161.502 & 130.436 & 53.923 & 130.032 & 40.94 \\
   \texttt{wiki} & 60 & 1521.62 & 121.401 & 146.298 & 59.783 & 139.855 & 42.321 \\
    \midrule
   \texttt{youtube} & 1 & 11292.7 & 3374.66 & 584.868 & 482.547 & 732.791 & 613.549 \\
   \texttt{youtube} & 2 & 6959.4 & 1680.46 & 397.548 & 286.539 & 455.884 & 339.003 \\
   \texttt{youtube} & 4 & 5074.02 & 1101.43 & 311.289 & 199.618 & 353.2 & 229.03 \\
   \texttt{youtube} & 8 & 3332.69 & 587.605 & 187.655 & 94.655 & 203.412 & 125.064 \\
   \texttt{youtube} & 16 & 2385.11 & 282.659 & 165.859 & 68.592 & 171.807 & 68.367 \\
   \texttt{youtube} & 30 & 1990.75 & 177.756 & 173.735 & 65.958 & 177.212 & 59.727 \\
   \texttt{youtube} & 60 & 2234.96 & 129.412 & 197.85 & 66.573 & 196.424 & 57.201 \\
    \midrule
\texttt{twitter}  & 60 & n/a & n/a & 10447.0 & 8413.4 & 9777.84 & 8509.94 \\
\texttt{friendster}  & 60 & n/a & n/a & 15357.8 & 10548.9 & 15150.3 & 9739.67 \\
\texttt{clueweb}  & 60 & n/a & n/a & 112645 & 83919.9 & 120970.0 & 87098.5 \\
\texttt{hyperlink2012}  & 60 & n/a & n/a & 352658 & 270398 & 375710 & 293445 \\

    \bottomrule

    \end{tabular}
    \end{center}
    \caption{Runtime for each graph on variants of our algorithms (in ms). Each algorithm is run for $20$ iterations except \ApproxGS where we run it for $100$ iterations.}
    \label{tab:graph_runtime_thread_algos}
    \end{table*}

\begin{table*}[h]
    \begin{center}
    \footnotesize
    \begin{tabular}{ll||rr|rr}
    \toprule
    Graph & Thread  &\PARGreedyPP & \PARSortingPP & \ApproxGP & \ApproxGS \\
    \midrule
   \texttt{brain} & 1 & 19138.2600 & 19077.6080 & 35608.3400 & 20482.6440 \\
   \texttt{brain} & 2 & 12511.9400 & 15238.0370 & 20504.3200 & 11283.8580 \\
   \texttt{brain} & 4 & 9233.2400 & 9154.3650 & 12835.6700 & 6547.5610 \\
   \texttt{brain} & 8 & 5449.5300 & 5445.0860 & 7186.9900 & 3449.3480 \\
   \texttt{brain} & 16 & 3402.8200 & 3393.6150 & 4265.0100 & 1982.5970 \\
   \texttt{brain} & 30 & 2615.7860 & 2634.6050 & 3097.7720 & 1323.6980 \\
   \texttt{brain} & 60 & 2655.9610 & 2665.0420 & 2868.9830 & 1074.6020 \\
    \midrule
   \texttt{dblp} & 1 & 134.1470 & 133.0390 & 166.6210 & 163.7430 \\
   \texttt{dblp} & 2 & 86.2590 & 86.1380 & 95.4420 & 93.9770 \\
   \texttt{dblp} & 4 & 63.3860 & 51.4390 & 62.1480 & 52.4300 \\
   \texttt{dblp} & 8 & 32.2350 & 32.3150 & 35.5430 & 30.7410 \\
   \texttt{dblp} & 16 & 21.5820 & 23.7550 & 22.4410 & 20.0650 \\
   \texttt{dblp} & 30 & 20.4660 & 22.4410 & 17.8620 & 16.5390 \\
   \texttt{dblp} & 60 & 24.6170 & 24.8650 & 18.9210 & 17.5910 \\
    \midrule
   \texttt{hepph} & 1 & 201.5330 & 202.1600 & 355.2360 & 202.3690 \\
   \texttt{hepph} & 2 & 165.5110 & 163.6800 & 246.7400 & 131.2910 \\
   \texttt{hepph} & 4 & 147.5400 & 109.7290 & 150.9970 & 73.6500 \\
   \texttt{hepph} & 8 & 109.0430 & 84.1890 & 101.0010 & 47.1910 \\
   \texttt{hepph} & 16 & 78.3790 & 72.5540 & 84.0600 & 27.4770 \\
   \texttt{hepph} & 30 & 90.8670 & 82.3670 & 86.0790 & 25.8150 \\
   \texttt{hepph} & 60 & 105.5090 & 98.5230 & 130.6320 & 25.1280 \\
    \midrule
   \texttt{lj} & 1 & 4698.8900 & 4656.7230 & 5889.5820 & 5733.1680 \\
   \texttt{lj} & 2 & 2761.8070 & 2730.5270 & 2817.9500 & 2759.1370 \\
   \texttt{lj} & 4 & 1968.6490 & 1968.1820 & 1756.3180 & 1720.9810 \\
   \texttt{lj} & 8 & 1169.4470 & 1127.4190 & 962.0660 & 897.7350 \\
   \texttt{lj} & 16 & 689.1800 & 685.9260 & 501.9480 & 476.5230 \\
   \texttt{lj} & 30 & 535.1070 & 526.6380 & 351.0520 & 316.0780 \\
   \texttt{lj} & 60 & 491.1170 & 490.9540 & 269.9850 & 243.3730 \\
    \midrule
   \texttt{orkut} & 1 & 10066.2290 & 9714.2330 & 13102.3690 & 11325.5780 \\
   \texttt{orkut} & 2 & 5734.3990 & 5796.1910 & 6324.2510 & 5231.8120 \\
   \texttt{orkut} & 4 & 4193.7490 & 4192.9590 & 3958.3140 & 3114.4700 \\
   \texttt{orkut} & 8 & 2496.5830 & 2468.8870 & 2182.2730 & 1702.6790 \\
   \texttt{orkut} & 16 & 1492.9540 & 1650.6160 & 1223.4590 & 887.7470 \\
   \texttt{orkut} & 30 & 1086.8080 & 1106.2830 & 849.1290 & 574.9270 \\
   \texttt{orkut} & 60 & 1033.0520 & 1027.3420 & 737.1380 & 435.4910 \\
    \midrule
   \texttt{stackoverflow} & 1 & 2826.2500 & 2838.8500 & 3846.5220 & 3248.5210 \\
   \texttt{stackoverflow} & 2 & 1693.1140 & 1688.1870 & 2044.4110 & 1649.1120 \\
   \texttt{stackoverflow} & 4 & 1189.3280 & 1197.3740 & 1310.6420 & 976.9350 \\
   \texttt{stackoverflow} & 8 & 705.7360 & 691.5880 & 707.8910 & 523.7900 \\
   \texttt{stackoverflow} & 16 & 404.6680 & 404.9710 & 412.2330 & 278.3090 \\
   \texttt{stackoverflow} & 30 & 318.5510 & 305.6890 & 319.8430 & 205.1740 \\
   \texttt{stackoverflow} & 60 & 302.8080 & 305.3290 & 272.3660 & 148.5310 \\
    \midrule
   \texttt{wiki} & 1 & 345.3700 & 342.8610 & 429.8260 & 396.4320 \\
   \texttt{wiki} & 2 & 215.3720 & 215.1640 & 256.0070 & 223.2740 \\
   \texttt{wiki} & 4 & 123.8130 & 144.7140 & 157.7860 & 130.8230 \\
   \texttt{wiki} & 8 & 76.9270 & 85.6690 & 83.8040 & 65.2850 \\
   \texttt{wiki} & 16 & 53.2990 & 53.7960 & 58.8050 & 39.4360 \\
   \texttt{wiki} & 30 & 52.7480 & 52.2590 & 52.3950 & 31.3690 \\
   \texttt{wiki} & 60 & 56.8670 & 56.9740 & 50.0650 & 27.0120 \\
    \midrule
   \texttt{youtube} & 1 & 469.5060 & 473.8470 & 616.2400 & 552.6010 \\
   \texttt{youtube} & 2 & 282.2730 & 281.0320 & 342.0240 & 297.2900 \\
   \texttt{youtube} & 4 & 194.1850 & 192.7370 & 237.8010 & 187.6070 \\
   \texttt{youtube} & 8 & 93.4370 & 91.5950 & 108.9290 & 95.9710 \\
   \texttt{youtube} & 16 & 67.9350 & 64.7360 & 74.4110 & 51.3480 \\
   \texttt{youtube} & 30 & 64.3360 & 61.9730 & 64.6510 & 42.2600 \\
   \texttt{youtube} & 60 & 63.2840 & 62.9000 & 63.0370 & 37.0870 \\
    \midrule
\texttt{twitter} & 60 & 7742.9950 & 6579.0130 & 6942.9250 & 6159.4840 \\
\texttt{friendster} & 60 & 11432.3900 & 10407.3190 & 13534.1800 & 8243.0770 \\
\texttt{clueweb} & 60 & 83446.3840 & 82912.9900 & 90663.6150 & 84847.2550 \\
\texttt{hyperlink2012} & 60 & 269341.7400 & 271172.1140 & 293033.4600 & 285467.1890 \\

    \bottomrule
    \end{tabular}
    \end{center}
    \caption{Initialization time for each graph on variants of our algorithms (in ms).}
    \label{tab:graph_inittime_thread_algos}
    \end{table*}

\clearpage
\section{Additional Figures from Section~\ref{sec:exp}}
\label{app:sec:figures}
\begin{figure*}[h!]
\centering
\includegraphics[width=1\textwidth]{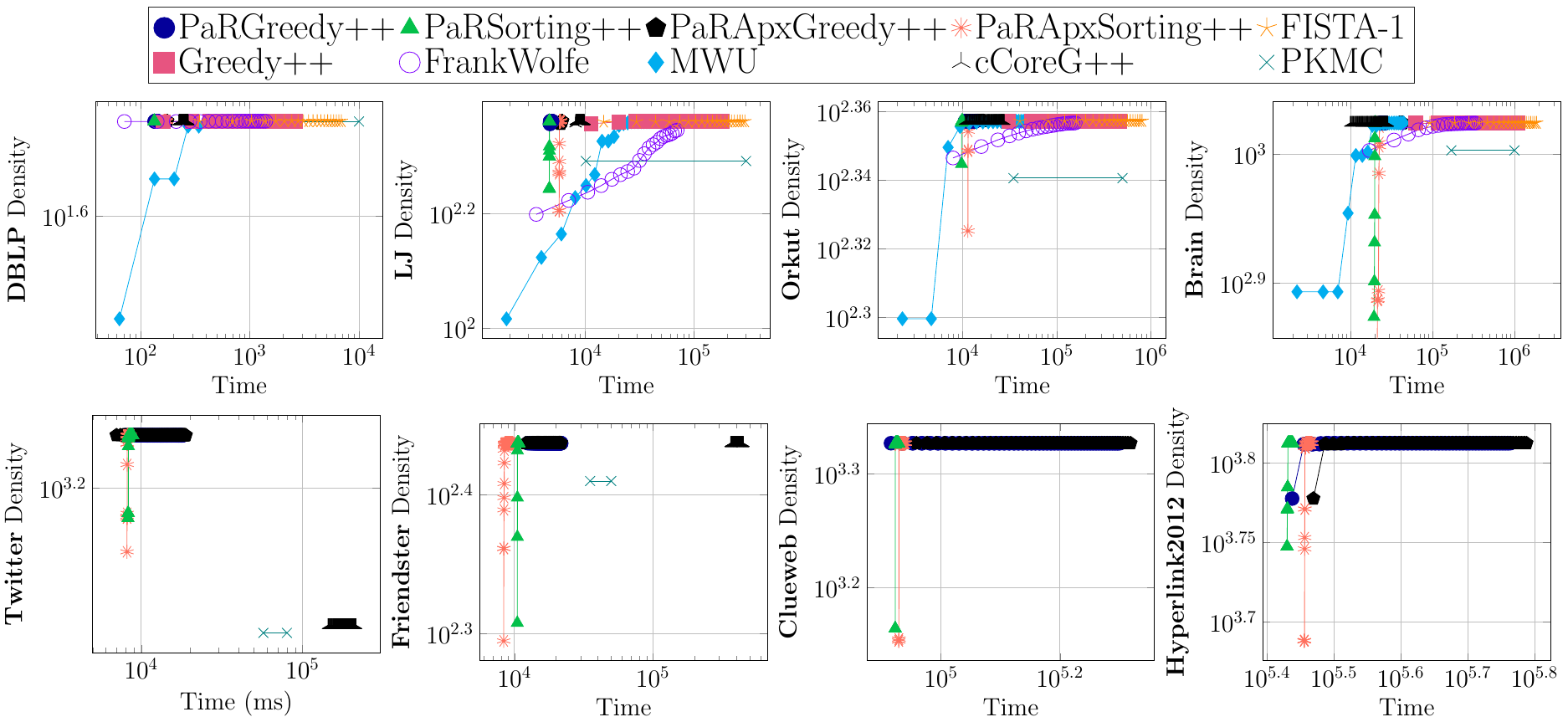}
\caption{Densities and time (ms) for various algorithms. \PARGreedyPP, \PARSortingPP, and \FISTA are run on one thread here for smaller graphs (top row) and $60$ threads for the large graphs (bottom row). 
}\label{fig:plot_den_time}
\end{figure*}

\begin{figure*}[h!]
\centering
\includegraphics[width=1\textwidth]{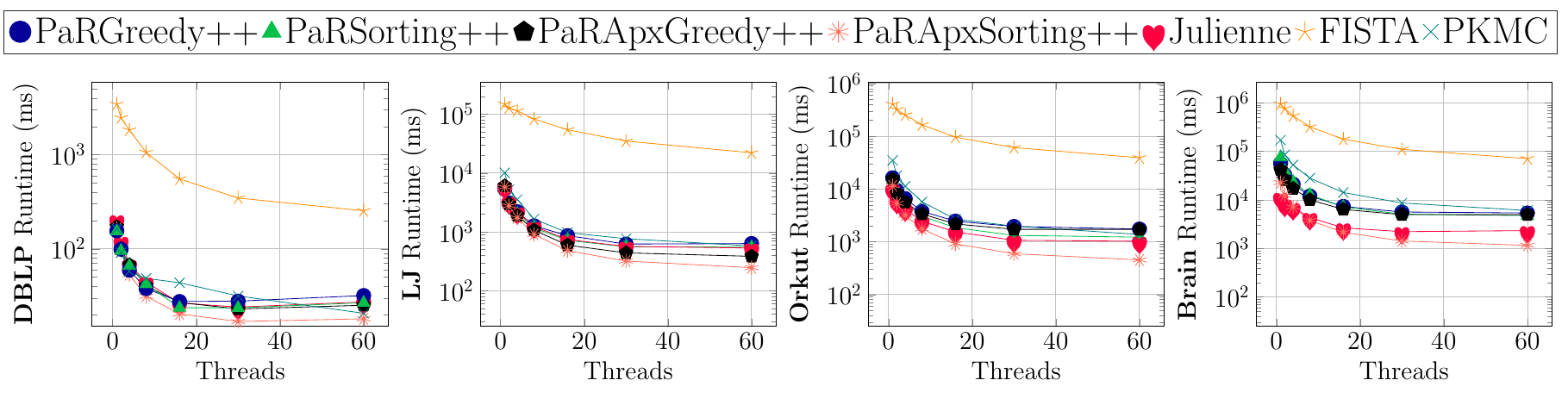}
\caption{Runtimes (ms) of \PARGreedyPP, \PARSortingPP, \ApproxGP, \ApproxGS, \julienne, \FISTA, and \pkmc versus the number of threads when running for $10$ iterations. }
\label{fig:threads-10}
\end{figure*}

\begin{figure*}[h!]
\centering
\includegraphics[width=1\textwidth]{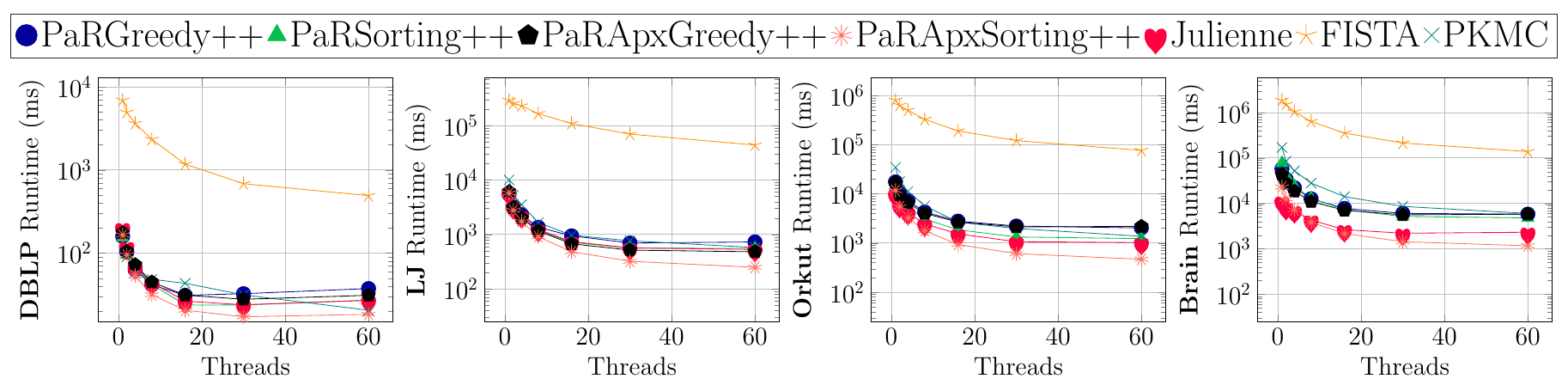}
\caption{Runtimes (ms) of \PARGreedyPP, \PARSortingPP, \ApproxGP, \ApproxGS, \julienne, \FISTA, and \pkmc versus the number of threads when running for $20$ iterations. }
\label{fig:threads-20}
\end{figure*}

\end{document}